\newif\ifgraphics\graphicsfalse
\graphicstrue

\documentclass[twocolumn,notitlepage,superscriptaddress,nofootinbib,10pt,a4paper,accepted=2026-03-22]{quantumarticle}

\pdfoutput=1

\usepackage{amssymb,amsmath,amscd,amsthm}
\usepackage{verbatim,xcolor,url}
\usepackage[utf8]{inputenc}
\usepackage{mathrsfs}
\usepackage{bbm}
\usepackage{graphicx}
\usepackage[numbers,sort&compress]{natbib}
\usepackage{enumitem}
\usepackage{hyperref}

\allowdisplaybreaks

\usepackage{mathtools}

\usepackage{pgfplots}
\usepackage{tikz}
\usetikzlibrary{matrix}
\usetikzlibrary{decorations.pathreplacing}
\usetikzlibrary{decorations.markings}
\usetikzlibrary{decorations.pathmorphing}
\usetikzlibrary{decorations.fractals}
\usetikzlibrary{shapes,shapes.callouts,arrows,positioning,patterns,patterns.meta,arrows.meta}
\usetikzlibrary{tikzmark}
\usetikzlibrary{calc}
\usetikzlibrary{fit}
\usetikzlibrary{fadings}
\usetikzlibrary{math}
\usetikzlibrary{spath3}
\tikzset{>=latex}

\pgfmathdeclarefunction{lambdaprime}{1}{%
	\pgfmathparse{sqrt(1-pow(#1,2))}
}
\pgfmathdeclarefunction{l0}{2}{%
	\pgfmathparse{ln((#2*(1+lambdaprime(#1)))/((1+lambdaprime(#2))*#1))}
}
\pgfmathdeclarefunction{l0dual}{2}{%
	\pgfmathparse{l0(#2,#1)}
}
\pgfmathdeclarefunction{lambdatwocrit}{1}{%
	\pgfmathparse{(2*exp(2*pi*abs(#1)))/(1+exp(4*pi*abs(#1)))}
}

\pgfmathdeclarefunction{lambdatwo}{2}{%
	\pgfmathsetmacro\lambdatwocrit{lambdatwocrit(#2)}%
	\ifdim #1 pt<\lambdatwocrit pt%
	\pgfmathparse{(2*exp(2*pi*abs(#2))*#1*(1+lambdaprime(#1)))/(2*(1+lambdaprime(#1))+pow(#1,2)*(exp(4*pi*abs(#2))-1))}%
	\else
	\pgfmathparse1%
	\fi
}
\pgfmathdeclarefunction{lyap}{2}{%
	\pgfmathparse{max(0,l0(#1,#2)}
}
\pgfmathdeclarefunction{arcsinh}{1}{%
	\pgfmathparse{ln(#1 + sqrt(pow(#1,2)+1)}
}
\pgfmathdeclarefunction{epszero}{1}{%
	\pgfmathparse{arcsinh(lambdaprime(#1)/#1)/2/pi}
}
\pgfmathdeclarefunction{lyapEps}{3}{
	\pgfmathparse{max(0,l0(#1,#2)/2/pi+abs(#3)-max(0,abs(#3)-epszero(#2)))}
}

\usepackage{tikzscale}
\pgfplotsset{compat=newest}

\usepackage{calc}

\newlength{\myeqwidth}
\DeclareMathOperator\tr{Tr}

\definecolor{purple}{rgb}{.5,0,1}
\definecolor{orange}{rgb}{1,.5,0}
\definecolor{pink}{rgb}{1,0,.5}
\definecolor{green}{rgb}{0,.5,0}
\definecolor{gold}{rgb}{1,.623,0}
\definecolor{kardinal}{RGB}{196,30,58}

\definecolor{myblue}{RGB}{100,100,220}
\definecolor{myred}{RGB}{255,100,100}
\definecolor{mygreen}{RGB}{119,221,119}
\definecolor{myorange}{RGB}{255,200,100}

\newcommand{\bbC}{{\mathbb{C}}}
\newcommand{\bbD}{{\mathbb{D}}}

\newcommand{\bbL}{{\mathbb{L}}}
\newcommand{\bbM}{{\mathbb{M}}}
\newcommand{\bbN}{{\mathbb{N}}}

\newcommand{\bbR}{{\mathbb{R}}}
\newcommand{\bbS}{{\mathbb{S}}}
\newcommand{\bbT}{{\mathbb{T}}}

\newcommand{\bbZ}{{\mathbb{Z}}}
\newcommand{\calL}{{\mathcal{L}}}
\newcommand{\calM}{{\mathcal{M}}}
\newcommand{\calO}{{\mathcal{O}}}
\newcommand{\calP}{{\mathcal{P}}}
\newcommand{\calT}{{\mathcal{T}}}
\newcommand{\calPT}{{\mathcal{PT}}}
\newcommand{\SL}{{\bbS\bbL}}

\newtheorem{theorem}{Theorem}
\newtheorem{lemma}[theorem]{Lemma}
\newtheorem{prop}[theorem]{Proposition}
\newtheorem{coro}[theorem]{Corollary}

\newtheorem{example}[theorem]{Example}

\newtheorem{remark}[theorem]{Remark}

\newcommand{\idty}{\mathbbm 1}
\newcommand{\aubrydual}{{\sharp}}
\newcommand{\realified}{{\mathbb{R}}}

\newcommand{\adj}{\dagger}

\def\wind{\nu}
\newcommand\eps{\varepsilon}
\newcommand\dos{\zeta} 
\newcommand\nr{n}
\newcommand\dr{m}

\newcommand\Psym{{\calP}}
\newcommand\Tsym{{\calT}}
\newcommand\PT{{\calPT}}

\makeatletter
\def\@printtitle{%
	{%
		\iftoggle{@unpublished}%
		{%
			\@printtitletextwithappropriatefontsize%
		}%
		{%
			\saveexploremode\exploregroups\StrSubstitute{\@title}{\newline}{}[\@titletemp]\restoreexploremode%
			\xdef\titleplain{\@titletemp}
			\edef\@titleexpanded{\detokenize\expandafter{\@titletemp}}%
			\iftoggle{@xstring}%
			{\saveexploremode\exploregroups\StrSubstitute{\@titleexpanded}{ }{\%20}[\@titleforurl]\restoreexploremode}%
			{\gdef\@titleforurl{\@titleexpanded}}%
			\href{https://quantum-journal.org/?s=\@titleforurl\&reason=title-click}{%
				\color{quantumviolet}{%
					\@printtitletextwithappropriatefontsize\unskip%
				}%
			}%
		}%
	}%
}
\AtBeginDocument{
	\iftoggle{@accepted}
	{%
		\savebox\@quantumacceptedbox{\textcolor{quantumviolet}{ \sffamily\footnotesize Accepted in {\normalsize\Quantum}\ifcsdef{@accepteddate}{ \@accepteddate}{}, click title to verify. Published under CC BY-NC-ND 4.0.}}}
	{}
}
\makeatother

\begin{document}

\title{Mobility edges in pseudo-unitary \newline quasiperiodic quantum walks}

\author{Christopher Cedzich}
\affiliation{Heinrich Heine Universit\"at D\"usseldorf, Universit\"atsstr. 1, 40225 D\"usseldorf, Germany}
\address{Fakult\"at f\"ur Mathematik und Informatik, FernUniversit\"at in Hagen, Universit\"atsstr. 1, 58097 Hagen, Germany}
\email{christopher.cedzich@fernuni-hagen.de}
\author{Jake Fillman}
\affiliation{Department of Mathematics, Texas A\&M University, College Station, TX, USA}
\email{fillman@tamu.edu}


\begin{abstract}
We introduce a Floquet quasicrystal that simulates the motion of Bloch electrons in a homogeneous magnetic field in discrete time steps. 
We admit the hopping to be non-reciprocal which, via a generalized Aubry duality, leads us to push the phase that parametrizes the synthetic dimension off of the real axis.
This breaks unitarity, but we show that the model is still ``pseudo-unitary''. 
We unveil a novel mobility edge between a metallic and an insulating phase that sharply divides the parameter space.
Moreover, for the first time, we observe a second transition that appears to be unique to the discrete-time setting. 
We quantify both phase transitions and relate them to properties of the spectrum.
If the hopping is reciprocal either in the lattice direction or the synthetic dimension, the model is $\PT$-symmetric, and the spectrum is confined to the unit circle up to a critical point. At this critical point, $\PT$-symmetry is spontaneously broken and the spectrum leaves the unit circle. This transition is topological and measured by a spectral winding number.
\end{abstract}

\maketitle

\hypersetup{
	pdftitle = {\titleplain}
}


\section{Introduction}
In introductory quantum mechanics, one is taught that Hamiltonians are hermitian, wherefore their spectra are real, and the converse is tacitly implied. 
This is challenged by theories that model measurable physical phenomena but defy a hermitian description while maintaining a real spectrum \cite{PhysRev.115.1390, BROWER1978213, CHARMS1980291, PhysRevLett.40.1610, PhysRevLett.54.1354, CARDY1989275, cmp/1103908052, hatanoLocalizationTransitionsNonHermitian1996, PhysRevB.56.8651}.
Vice versa, one can construct non-hermitian Hamiltonians through extending parameter ranges or modifying the potential while preserving reality of the spectrum \cite{hatanoLocalizationTransitionsNonHermitian1996, hatanoNonHermitianDelocalizationEigenfunctions1998, longhiBlochOscillationsComplex2009, sarnakSpectralBehaviorQuasi1982}.
In a first attempt to resolve this ostensible paradox, it was realized that such Hamiltonians often posses a ``parity-time'' ($\PT$) symmetry \cite{benderRealSpectraNonHermitian1998, benderPTsymmetricQuantumMechanics1999, benderIntroductionPTSymmetricQuantum2005}.
Soon after, it became clear that $\PT$-symmetry is a special case of the more general pseudo-hermiticity \cite{zhangPTsymmetryEntailsPseudoHermiticity2020, mostafazadehPseudoHermiticityPTSymmetry2002, mostafazadehPseudohermitianRepresentationQuantum2010}, and 
that positivity of the corresponding pseudo-metric guarantees the main physical requirement on the time-evolution, namely the preservation of probabilities 
\cite{mostafazadehPseudoHermiticityPTSymmetry2002,mostafazadehPseudoHermiticityPTsymmetryII2002,mostafazadehPseudohermitianRepresentationQuantum2010}.

Since its discovery, pseudo-hermitian and $\PT$-symmetric quantum mechanics have remained highly active areas of research.
Recent work has revealed even more rich phenomena, such as the non-Hermitian skin effect \cite{linTopologicalPhaseTransitions2022, yaoEdgeStatesTopological2018b, yaoNonHermitianChernBands2018, liuLocalizationTransitionSpectrum2021a}, non-Hermitian edge burst \cite{xiaoObservationNonHermitianEdge2023, Xue_2022}, novel topological phases \cite{bernardClassificationNonHermitianRandom2002, gongTopologicalPhasesNonHermitian2018, zhouPeriodicTableTopological2019,xiaoNonHermitianBulkBoundary2020}, and non-hermitian quasicrystals  \cite{liuLocalizationTransitionSpectrum2021a,wangSpectrumHatanoNelsonModel2023, yucePTSymmetricAubry2014,longhiBlochOscillationsComplex2009,longhiNonBlochCalCal2019,longhiTopologicalPhaseTransition2019}.

This motivates an investigation of similar phenomena for other physical systems. 
A rich class of models is furnished by quantum walks, which recently attracted much attention as a computational resource \cite{Ambainis, Lovett:2010ff, portugal2013quantum,Shenvi:2003be,santha_QW_search} as well as an excellent testbed for discrete-time quantum lattice dynamics \cite{ASW2011JMP,AVWW2011JMP, SpacetimeRandom,Joye_Merkli, aschLowerBoundsLocalisation2019, joye_d_dim_loc,farrellyReviewQuantumCellular2020, arrighiOverviewQuantumCellular2019, schumacherReversibleQuantumCellular2004} and topological phases \cite{WeAreSchur, F2W,Ti, TopClass, asbothBulkBoundaryCorrespondence2013, KitaExploring, asch2019engineeringstablequantumcurrents}. 
Moreover, quantum walks have been demonstrated experimentally in optical lattices \cite{OpticalLattice1, OpticalLattice2, ElektricExp}, trapped ions \cite{TrappedIons1,TrappedIons2}, optical fibers \cite{OpticalFibers1,OpticalFibers2}, photonic circuits \cite{PhotonicCircuits1} and photonic waveguide arrays \cite{WaveGuide1,WaveGuide2,WaveGuide3}.

Also in the $\PT$-symmetric regime, quantum walks have already been studied theoretically \cite{mochizukiExplicitDefinitionPT2016,mochizukiBulkedgeCorrespondenceNonunitary2020,asaharaIndexTheoremOnedimensional2021a,longhiNonBlochCalCal2019,hatanoDelocalizationNonHermitianQuantum2021} and successfully implemented experimentally \cite{xiaoNonHermitianBulkBoundary2020,regensburgerParityTimeSynthetic2012a,regensburgerObservationDefectStates2013,wimmerObservationOpticalSolitons2015}. In particular, the experimental demonstration of topologically protected localization-delocalization transitions in \cite{linTopologicalPhaseTransitions2022} and the triple phase transition in \cite{weidemannTopologicalTriplePhase2022} demonstrates that $\PT$-symmetric quasiperiodic operators are implementable in the lab. However, it shall be noted that most of these implementations retract to a detour via the effective Hamiltonian to study the manifestations of $\PT$-symmetry instead of establishing the presence of $\PT$-symmetry directly for the time-evolution operator itself.

In this work, we obtain a new class of physically realizable pseudo-unitary quantum walks. Concretely, we study the \emph{unitary almost-Mathieu operator} (UAMO) introduced in \cite{CFO1} as a one-dimensional quantum walk model that derives from a two-dimensional quantum walk in a magnetic field, with a rich phase diagram exhibiting a metal-insulator transition with an intermediate critical regime of exotic singular continuous Cantor spectrum. 

We incorporate a non-reciprocal hopping $\eta$ that functions as a gain-loss parameter in the lattice direction. Applying Aubry duality in this regime leads us to introduce a complex phase $\eps$, which is naturally interpreted as a gain-loss parameter in the synthetic direction.
We analyze the phase diagram in $\eta$ and $\eps$ of this pseudo-unitary almost Mathieu operator (PUAMO) and observe metal-insulator transitions
that are characterized in terms of complexified Lyapunov exponents. We observe a stability of the spectrum in $\eta$ and $\eps$ separately below a certain critical value that we quantify as the Lyapunov exponent of the Aubry dual model. 
Along the principal axes of the $(\eta,\eps)$-plane, where the model is $\PT$-symmetric, this implies that the spectrum lies completely on the unit circle below the critical value. At the critical value, $\PT$-symmetry is spontaneously broken and some of the eigenvalues leave the unit circle. 

Unique to the discrete time setup, we unveil the existence of a novel second phase transition. 
The corresponding critical values $\eta_0$ and $\eps_0$ depend solely on the coupling in lattice direction and synthetic dimension, respectively. At these values, either the system decouples or the non-reciprocal gain of the shift trumps any localization capacity. Beyond the critical values, there are no eigenvalues left on the unit circle.
Moreover, $\eta_0$ and $\eps_0$ establish maximum values for the mobility edges: whenever the complexified Lyapunov exponents surpass $\eta_0$ or $\eps_0$, the metal-insulator transition takes place at these critical values and constancy of the spectrum holds up to there.
We observe another striking feature: 
$\eta$ and $\eps$ break duality in the sense that for some choices of these parameters, the phase of the PUAMO does not change upon dualization. Concretely, there are choices of $\eta$ and $\eps$ for which both the PUAMO as well as its Aubry-dual exhibit transport. This is a direct consequence of the second phase transition and therefore does not occur in continuous time.


\section{Setting}
\label{sec:model}
We consider particles on the one-dimensional lattice $\bbZ$ with two-dimensional cells $\bbC^2$ attached at each site that evolve under the repeated action of the split-step quantum walk $W=S_\lambda Q$. The state-dependent shift is defined by $
\big(S_\lambda\psi)_{n}^{\pm}
=\lambda \psi_{n\mp 1}^{\pm} \mp\lambda'\psi_{n}^{\mp}
$,
where $\lambda'=\sqrt{1-\lambda^2}$, $n\in\bbZ$ labels the position and $\pm$ the index of the local degree of freedom. 
The coin operation $Q$ acts locally via a matrix $Q_n\in \mathrm{SU}(2)$ which we parametrize as
\begin{equation*}
	Q_n=\begin{bmatrix}q_n^{11}&q_n^{12}\\q_n^{21}&q_n^{22}\end{bmatrix}
\end{equation*}
see Figure \ref{fig:walk_action}.

\begin{figure}[t]
	\begin{center}
		\begin{tikzpicture}[baseline={($ (current bounding box.center) + (0,.05) $)},scale=.8]
	\def\xdist{1.5}
	\def\ydist{.7}
	\tikzstyle{dot} =[circle,fill,inner sep = 1.8]
	\tikzstyle{rects} = [rectangle, rounded corners=6, very thick, draw=red, fill=red, fill opacity=.2,inner sep=1.8mm, label={[above=.3cm]:#1}]
	\tikzstyle{rectstriv} = [rectangle, rounded corners=6, very thick, draw=gray!50, dashed, inner sep=1.8mm, label={[above=.4cm]:#1}]
	\tikzstyle{arr} = [->,thick,black,>=latex]
	\tikzstyle{l1} = [midway,auto,font=\footnotesize]
	\foreach \x [evaluate] in {0,1,2}{
		\foreach \y in {1,2}{		
			\node[dot] (\x\y) at (\xdist*\x,+\ydist-\ydist*\y) {};
		}
	}
	\node[rects={$Q_{n-1}$},fit=(01) (02)] {};
	\node[rects={$Q_{n}$},fit=(11) (12)] {};
	\node[rects={$Q_{n+1}$},fit=(21) (22)] {};
	\draw[dotted,thick] (-1.2*\xdist,+.5*\ydist-\ydist) -- (-.8*\xdist,+.5*\ydist-\ydist);
	\draw[dotted,thick] (2.8*\xdist,+.5*\ydist-\ydist) -- (3.2*\xdist,+.5*\ydist-\ydist);
	\node[inner sep=0,minimum width=.15cm] (leftup) at (-\xdist,\ydist-1*\ydist) {};
	\node[inner sep=0,minimum width=.15cm] (rightup) at (3*\xdist,\ydist-1*\ydist) {};
	\node[inner sep=0,minimum width=.15cm] (leftdown) at (-\xdist,\ydist-2*\ydist) {};
	\node[inner sep=0,minimum width=.15cm] (rightdown) at (3*\xdist,\ydist-2*\ydist) {};
	\draw[arr] (leftup) to[out=30,in=150] (01);
	\draw[arr] (01) to[out=30,in=150] (11);
	\draw[arr] (11) to[out=30,in=150] node[l1] {$\lambda$} (21);
	\draw[arr] (21) to[out=30,in=150] (rightup);
	\draw[arr] (02) to[out=-150,in=-30] (leftdown);
	\draw[arr] (12) to[out=-150,in=-30] node[l1] {$\lambda$} (02);
	\draw[arr] (22) to[out=-150,in=-30] (12);
	\draw[arr] (rightdown) to[out=-150,in=-30] (22);
	
	\draw[arr] (12) to[out=140,in=-130] node[l1,left] {$-\lambda'$} (11);
	\draw[arr] (11) to[out=-40,in=50] node[l1,right] {$\lambda'$} (12);
	
	\foreach \x in {0,2}{
		\draw[arr] (\x2) to[out=140,in=-130] (\x1);
		\draw[arr] (\x1) to[out=-40,in=50] (\x2);
	}

\end{tikzpicture}
	\end{center}
	\caption{\label{fig:walk_action}The action of the split-step walk $W=S_\lambda Q$. The arrows indicate the action of the shift $S_{\lambda}$, where for the sake of clarity the parameters are displayed only at a single lattice site. In the red local ``cells'', the respective coin $Q_n$ is acting.}
\end{figure}
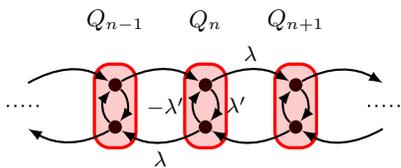

We take $Q_{n}=Q_{\lambda_2,\Phi,\theta, n}$ to be the quasiperiodic matrix
\begin{equation*}
	\resizebox{\columnwidth}{!}{$
		Q_n :=
		\begin{bmatrix}\lambda_2\cos(2\pi (n\Phi+\theta))+i\lambda_2'&-\lambda_2\sin(2\pi (n\Phi+\theta))\\\lambda_2\sin(2\pi (n\Phi+\theta))&\lambda_2\cos(2\pi 	(n\Phi+\theta))-i\lambda_2'\end{bmatrix},
		$}
\end{equation*}
where $\Phi \in [0,1]$ plays the role of a \emph{magnetic field} in an associated two-dimensional system and the \emph{phase} $\theta \in [0,1]$ is the Fourier parameter of a synthetic dimension \cite{CFGW2020LMP, CFO1}.
We follow \cite{CFO1} in calling the overall time step
\begin{equation*}
	W_{\lambda_1,\lambda_2,\Phi,\theta}=S_{\lambda_1}Q_{\lambda_2,\Phi,\theta}
\end{equation*}
the \emph{unitary almost-Mathieu operator} (UAMO) due to its close resemblance to the celebrated almost-Mathieu operator (AMO)\footnote{The AMO is also known as ``Aubry-André-'' or ``Harper-model''.} \cite{BelSim1982JFA,AJ2009Ann,Jitomirskaya1999Annals,papillon}. We often abbreviate the UAMO as $W_{\lambda_1,\lambda_2}$ when $\Phi$ and $\theta$ are given.
The parameters $\lambda_1,\lambda_2\in[0,1]$ are \emph{coupling constants}: 
$\lambda_1$ controls the strength of the shift $S_{\lambda_1}$ in lattice direction, whereas $\lambda_2$ parametrizes a shift in a synthetic dimension \cite{CFGW2020LMP, CFO1}: viewing $W_{\lambda_1 ,\lambda_2}$ as the implementation of a two-dimensional quantum walk in a magnetic field, we can understand multiplication by $e^{2\pi i\theta}$ as a Fourier-transformed shift with Fourier variable $\theta$. The strength of this shift is controlled by $\lambda_2$. In the limit cases $\lambda_1 = 0$ and $\lambda_1 = 1$ the shift $S_{\lambda_1}$ either decouples the walk into a direct sum of blocks without any transport or becomes the standard state-dependent shift, respectively. On the other hand, for $\lambda_2 = 0$ the coin is a constant diagonal matrix, while for $\lambda_2 = 1$ it corresponds to a rotation matrix with position-dependent angle \cite{FOZ2017CMP}.

For $\Phi=\nr/\dr$ rational, the model is periodic and initial states propagate ballistically under repeated application of $W_{\lambda_1,\lambda_2}$ 
\cite{AVWW2011JMP,extail,damanikSpreadingEstimatesQuantum2016}. Moreover, the spectrum of $W_{\lambda_1,\lambda_2}$ consists of $2\dr$ bands with Bloch waves as eigenstates.

For the rest of the paper, we consider only irrational fields. In this case, the UAMO exhibits a metal-insulator transition in its phase diagram about the critical line $\lambda_1=\lambda_2$ \cite{CFO1}, see Figure \ref{fig:phase_diagram}. This can be intuitively understood: for $\lambda_1>\lambda_2$ the shift dominates and one observes ballistic transport, whereas for $\lambda_1<\lambda_2$ the coin dominates, and the walk localizes. In analogy with the AMO we call these regimes ``subcritical'' and ``supercritical'', respectively. In the ``critical'' regime $\lambda_1=\lambda_2$, the system displays anomalous transport \cite{CFO1}. 

\def\circlerad{.8ex}
\def\circleshift{-.7ex}
\newcommand\Square[1]{+(-#1,-#1) rectangle +(#1,#1)}
\def\colorsquare#1{\tikz[baseline=\circleshift]\draw[#1,fill=#1,rotate=45] (0,0) \Square{\circlerad};}

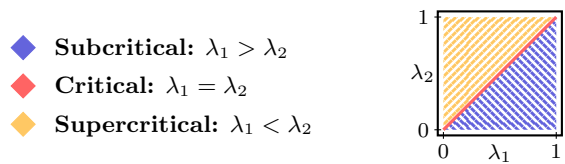
\begin{figure*}[t]
	\begin{center}
	\begin{minipage}{.25\textwidth}
		\begin{description}
			\item[\colorsquare{myblue}\ \ Subcritical] 
			\item[\colorsquare{myred}\ \ Critical] 
			\item[\colorsquare{myorange}\ \ Supercritical] 
		\end{description}
	\end{minipage}
	\begin{minipage}{.25\textwidth}
		\begin{tikzpicture}
	[
	scale=1.7,
	font=\footnotesize
	]
	
	\node at (-.5,1.15) {$(a)$};
	
	\draw[thick,black] (-.05,-.05) rectangle +(1.1,1.1);
	
	\foreach \i in {0,1}{
		\draw[align=left] (-.02,{\i}) -- (-.08,{\i});
		\draw[align=left] ({\i},-.02) -- ({\i},-.08);
	}
	
	\draw (-.05,0) node[left, align=left] {$0$};
	\draw (-.05,1) node[left, align=left] {$1$};
	\draw (0,-.05) node[below, align=center] {$0$};
	\draw (1,-.05) node[below, align=center] {$1$};
	
	\draw (0,.5) node[left,align=right] {$\lambda_2$};
	\draw (.5,-.05) node[below,align=center] {$\lambda_1$};
	
	\path[pattern={Lines[angle=-45, distance={3pt/sqrt(2)}, line width= 1.1pt]}, pattern color=myblue] (0,0) -- (1,1) -- (1,0) -- cycle;
	\path[preaction={fill=white}, pattern={Lines[angle=45, distance={3pt/sqrt(2)}, line width= 1.1pt]}, pattern color=myorange] (0,0) -- (1,1) -- (0,1) -- cycle;
	
	\draw[line width=1.2pt, myred] (0,0) -- (1,1);
	
\end{tikzpicture}
	\end{minipage}
	\begin{minipage}{.25\textwidth}
		\def\EtaPhaseDiagram{.1}
		\begin{tikzpicture}
	[
	scale=1.7,
	font=\footnotesize
	]
	
	\node at (-.5,1.15) {$(b)$};
	
	\draw[thick,black] (-.05,-.05) rectangle +(1.1,1.1);
	
	\foreach \i in {0,1}{
		\draw[align=left] (-.02,{\i}) -- (-.08,{\i});
		\draw[align=left] ({\i},-.02) -- ({\i},-.08);
	}
	
	\draw (-.05,0) node[left, align=left] {$0$};
	\draw (-.05,1) node[left, align=left] {$1$};
	\draw (0,-.05) node[below, align=center] {$0$};
	\draw (1,-.05) node[below, align=center] {$1$};
	
	\draw (0,.5) node[left,align=right] {$\lambda_2$};
	\draw (.5,-.05) node[below,align=center] {$\lambda_1$};
	
	\path[spath/save=critLine, samples=100, domain=0:1] plot (\x,{lambdatwo(\x,\EtaPhaseDiagram)});
	
	\path[pattern={Lines[angle=-45, distance={3pt/sqrt(2)}, line width= 1.1pt]}, pattern color=myblue] [spath/use= critLine] -- (1,0) -- cycle;
	\path[pattern={Lines[angle=45, distance={3pt/sqrt(2)}, line width= 1.1pt]}, pattern color=myorange] [spath/use= critLine] -- (0,1) -- cycle;
	
	\draw[line width=1.2pt, myred] [spath/use=critLine];
	
%
%
	
\end{tikzpicture}
	\end{minipage}
	\end{center}
	\caption{\label{fig:phase_diagram}(a) The (spectral) phase diagram of the unitary almost-Mathieu operator \cite{CFO1}. The system propagates ballistically in the subcritical regime $\lambda_1>\lambda_2$ whereas it displays Anderson localization in the Aubry-dual supercritical regime $\lambda_1<\lambda_2$. In the critical regime $\lambda_1=\lambda_2$, the system shows anomalous transport.\\
	(b) The phase diagram of the pseudo-unitary almost-Mathieu operator with fixed non-reciprocal hopping $\eta=0.1$. In this setting, the different phases are determined by $L_{\lambda_1,\lambda_2,0,\eta}$.}
\end{figure*}

This behaviour suggests that a suitable ratio between $\lambda_1$ and $\lambda_2$ is analogous to that of the coupling constant of the AMO. 
The precise analogy is fostered by 
\begin{equation*}
	\lambda_0=\frac{\lambda_2(1+\lambda_1')}{\lambda_1(1+\lambda_2')},
\end{equation*}
with $\lambda_0 \to 0$ as $\lambda_2\to 0$ and $\lambda_0 \to
\infty$ as $\lambda_1 \to 0$. Moreover, in the critical regime where
$\lambda_1 = \lambda_2$ one has $\lambda_0 = 1$.
Indeed, the Lyapunov exponent of the UAMO for $z\in\partial\bbD$ in the spectrum of $W_{\lambda_1,\lambda_2}$ is independent of $z$ and given by \cite{CFO1}
\begin{equation*}
	L_{\lambda_1,\lambda_2}(z)=\max\{0,\log\lambda_0\}.
\end{equation*}
It therefore vanishes in the subcritical and the critical regimes, and is positive in the supercritical regime. The Lyapunov exponent quantifies their decay in the sense that a typical eigenstate will behave like $|\psi_n^\pm| \sim e^{-L|n-n_0|}$ as $n \to \infty$; thus eigenstates are extended in the subcritical and critical regime, whereas in the supercritical regime they are localized.

Moreover, the sub- and the supercritical regimes are dual to each other via an Aubry transform. More precisely, $W_{\lambda_1,\lambda_2}$ is dual to
\begin{equation*} 
W_{\lambda_1,\lambda_2}^\aubrydual := W_{\lambda_2,\lambda_1}^\top,
\end{equation*}
where $\top$ denotes the transpose of an operator in the standard basis.
In particular, $W_{\lambda_1,\lambda_2}$ is self-dual in the critical regime. The physical interpretation is thus that for the dual operator, the roles of the coupling constants are interchanged, that is, for the dual model $W_{\lambda_1,\lambda_2}^\aubrydual$, $\lambda_2$ is the coupling constant in the lattice direction while $\lambda_1$ parametrizes a shift in the synthetic dimension.
The Lyapunov exponent $L^\aubrydual(z)$ of the dual model is given by
\begin{equation*}
	L^\aubrydual_{\lambda_1,\lambda_2}(z)=L_{\lambda_2,\lambda_1}(z)=\max\{0,-\log\lambda_0\}.
\end{equation*}
This turns out to be pivotal: 
the Lyapunov exponent of the dual model determines the points at which the slope of the original Lyapunov exponent changes \cite{geMultiplicativeJensenFormula2023}; see Figure \ref{fig:lyap_exp_UAMO}.

\section{Non-reciprocal hopping and complexified phase}
We now break the reciprocality of $S_\lambda$ by introducing a gain-loss parameter $\eta$ in the lattice direction,  viz.:
\begin{equation*}
	\big(S_{\lambda,\eta}\psi\big)_n^\pm
	= e^{\pm 2\pi \eta} \lambda \psi_{n \mp 1}^{\pm} \mp\lambda'\psi_{n}^{\mp}.
\end{equation*}
Depending on its sign, the role of $\eta$ is to damp transport in one direction and amplify it in the other.
Dualizing the walk with this shift, we obtain a walk with reciprocal shift but whose coin has a complexified quasiperiodic phase. This serves as a motivation to complexify the phase $\theta$ by letting
\begin{equation*}
	\theta\mapsto\vartheta:=\theta+i\eps,
\end{equation*}
and denoting the resulting coins by $Q_{\lambda_2,\eps}$.
We remark that the imaginary phase $\eps$ appears already in \cite{CFO1} as a technical means to compute the Lyapunov exponent $L_{\lambda_1,\lambda_2}(z)$ and distinguish the different phases of the UAMO. From the underlying two-dimensional model, it has the physical interpretation of a non-reciprocal hopping parameter in the synthetic dimension.

Incorporating these two novel parameters into the UAMO, we define the \emph{pseudo-UAMO} (PUAMO)
\begin{equation*}
	W_{\lambda_1,\lambda_2,\eta,\eps} = S_{\lambda_1,\eta}Q_{\lambda_2,\eps}.
\end{equation*}
By construction, the PUAMO possesses a \emph{generalized Aubry duality} between $W_{\lambda_1,\lambda_2,\eta,\eps}$ and
\begin{equation*}
	W_{\lambda_1,\lambda_2,\eta,\eps}^\aubrydual=W_{\lambda_2,\lambda_1,-\eps,-\eta}^\top
\end{equation*}
which {also} follows from incorporating $\eta$ and $\eps$ into the computations of \cite{CFO1}.
 
\def\loneSub{.5}
\def\ltwoSub{.25}
\def\loneCrit{.5}
\def\ltwoCrit{.5}
\def\loneSuper{.25}
\def\ltwoSuper{.5}
\begin{figure*}[t]
	\begin{minipage}[c][.18\textheight][t]{.87\textwidth}
		\tikz{\node at (0,0) {(a)};}
		\raggedright
		\vspace{-.2cm}
		\ifgraphics\def\xmax{.5}
\def\ymax{.5}

\def\legendposx{10*2*\xmax+.2}
\def\legendposy{0}
\def\circlerad{.8ex}
\def\circleshift{-.7ex}
\def\colorsquare#1{\tikz[baseline=\circleshift]\draw[#1,fill=#1,rotate=45] (0,0) \Square{\circlerad};}

\tikzset{
	pics/phase diagram UAMO eta eps/.style n args={3}{
		code={
						
			\draw[thick,black,scale=#3] ($(-\xmax,-\ymax)+(-.05,-.05)$) rectangle ($(\xmax,\ymax)+(.05,.05)$);
			
			\ifdim#1 pt>#2 pt%
				\draw[scale=#3] ($(0,\ymax)+(0,.1)$) node[font=\scriptsize, above, align=center] {$\lambda_1>\lambda_2$};
			\fi
			\ifdim#1 pt=#2 pt%
				\draw[scale=#3] ($(0,\ymax)+(0,.1)$) node[font=\scriptsize, above, align=center] {$\lambda_1=\lambda_2$};
			\fi
			\ifdim#1 pt<#2 pt%
				\draw[scale=#3] ($(0,\ymax)+(0,.1)$) node[font=\scriptsize, above, align=center] {$\lambda_1<\lambda_2$};
			\fi
			
			\draw[align=left,scale=#3] ($(-\xmax,0)+(-.08,0)$) -- node[font=\scriptsize, left, align=left] {$0$} ++(-.04,0);
			\draw[align=left,scale=#3] ($(0,-\ymax)+(0,-.08)$) -- node[font=\scriptsize, below, align=center] {$0$} ++(0,-.04);
			
			\foreach \i in {+}{
				\draw[align=left,scale=#3] ($(-\xmax,{\i1*epszero(#1)})+(-0.08,0)$) -- node[font=\scriptsize, left, align=left] {$\eta_0$} ++(-0.04,0);
				\draw[align=none,scale=#3] ($({\i1*epszero(#2)},-\ymax)+(0,-0.08)$) -- node[font=\scriptsize, below, align=center] {$\varepsilon_0$} ++(0,-0.04);
				
				\ifdim#1 pt<#2 pt 
					\draw[align=left,scale=#3] ($(-\xmax,{\i1*l0(#1,#2)/(2*pi)})+(-0.08,0)$) -- node[font=\scriptsize, left, align=left] {$\frac{L}{2\pi}$} ++(-0.04,0);
				\fi
				\ifdim#1 pt>#2 pt 
					\draw[align=none,scale=#3] ($({\i1*l0dual(#1,#2)/(2*pi)},-\ymax)+(0,-0.08)$) -- node[font=\scriptsize, below, align=center] {$\frac{L^\aubrydual}{2\pi}$} ++(0,-0.04);
				\fi
			}
			

			\draw[scale=#3] ($(-\xmax,\ymax)+(-.1,0.05)$) node[font=\scriptsize,left,align=right] {$\eta$};
			\draw[scale=#3] ($(\xmax,-\ymax)+(0.1,-.1)$) node[font=\scriptsize,below,align=center] {$\eps$};
			
			
			\path[spath/save= lyapEpsPlusMinus, domain=-\xmax:0, samples=50, scale= #3] plot (\x,{lyapEps(#1,#2,\x)});
			\path[spath/save= lyapEpsPlusPlus, domain=0:\xmax, samples=50, scale= #3] plot (\x,{lyapEps(#1,#2,\x)});
			\path[spath/save= lyapEpsMinusPlus, domain=0:\xmax, samples=50, scale= #3] plot (\x,{-lyapEps(#1,#2,\x)});	
			\path[spath/save= lyapEpsMinusMinus, domain=-\xmax:0, samples=50, scale= #3] plot (\x,{-lyapEps(#1,#2,\x)});
			
			\path[spath/save= lyapEpsPlus] [spath/use=lyapEpsPlusMinus] [spath/use={lyapEpsPlusPlus,weld}];
			\path[spath/save= lyapEpsMinus] [spath/use=lyapEpsMinusMinus] [spath/use={lyapEpsMinusPlus,weld}];
			
			\path[preaction={fill=white}, pattern=north west lines, pattern color=myblue, scale= #3] [spath/use=lyapEpsPlus] -- (\xmax,\ymax) -- (-\xmax,\ymax) -- cycle;
			\path[preaction={fill=white}, pattern=north west lines, pattern color=mygreen, scale= #3] [spath/use=lyapEpsMinus] -- (\xmax,-\ymax) -- (-\xmax,-\ymax) -- cycle;
			\path[preaction={fill=white}, pattern=north east lines, pattern color=myorange, scale= #3] [spath/use=lyapEpsPlusMinus] [spath/use={lyapEpsMinusMinus,reverse,weld}];
			\path[preaction={fill=white}, pattern=north east lines, pattern color=myorange, scale= #3] [spath/use=lyapEpsPlusPlus] [spath/use={lyapEpsMinusPlus,reverse,weld}];
			
			
			\draw[thick, kardinal, scale= #3] (-\xmax,0) -- (\xmax,0) (0,-\ymax) -- (0,\xmax);
			\draw[very thick, myred] [spath/use= lyapEpsPlus] [spath/use= lyapEpsMinus];
			
			\ifdim#1 pt<#2 pt 
			\draw[very thick, myred, dotted, scale= #3] (0,{lyapEps(#1,#2,0)}) -- (0,{-lyapEps(#1,#2,0)});
			\fi
		}
	}
}

\begin{tikzpicture}[scale=.95]
	\path (0,0) pic {phase diagram UAMO eta eps={\loneSub}{\ltwoSub}{2.3}};
	
	\path ($(7*\xmax+.5,0)$) pic {phase diagram UAMO eta eps={\loneCrit}{\ltwoCrit}{2.3}};
	
	\path ($(2*7*\xmax+1,0)$) pic {phase diagram UAMO eta eps={\loneSuper}{\ltwoSuper}{2.3}};
	\node[anchor=west,font=\footnotesize] at (\legendposx,\legendposy) {\begin{tabular}{c@{\hspace{.5em}}|@{\hspace{.5em}}ll}
											&	$L_{\eta,\eps}^{\mathrm{right}}$ & $L_{\eta,\eps}^{\mathrm{left}}$ \\\hline
					\colorsquare{myorange}	&	$>0$ & $>0$	\\
					\colorsquare{myblue}	&	$>0$ & $<0$	\\
					\colorsquare{mygreen}	&	$<0$ & $>0$
			\end{tabular}};
	
\end{tikzpicture}\fi
	\end{minipage}
	\begin{minipage}[c][.18\textheight][t]{.11\textwidth}
		\raggedright
		\tikz{\node at (0,0) {(b)};
				\node[anchor=north west] at (.1,-.3) {\includegraphics[height=.13\textheight]{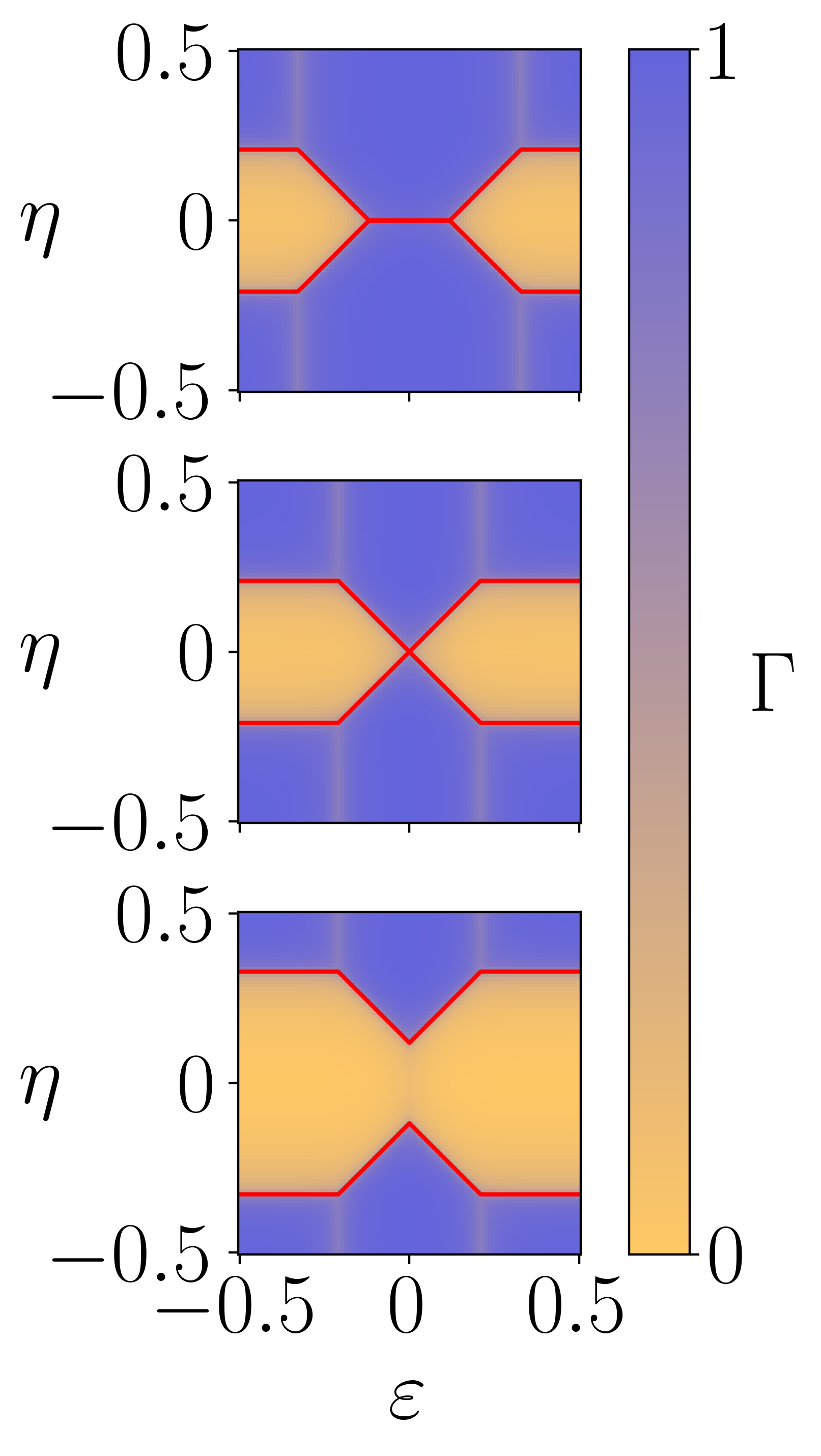}};
			}
	\end{minipage}
	\caption{\label{fig:eta_eps_phasediagram}
		(a) Phase diagrams of the PUAMO in $(\eps,\eta)$ with phases separated by $L_{\lambda_1,\lambda_2,\eta,\eps}^{\text{left/right}}=0$ (red): either there is transport to the right (blue) or the left (green), or the system displays Anderson localization (yellow). Only in the subcritical regime there is a parameter range ($\eta=0$, $0\leq|\eps|\leq L_{\lambda_1,\lambda_2,0,0}^\aubrydual/(2\pi)$) for which both the left and right Lyapunov exponents vanish and the PUAMO delocalizes in both directions. 
		Along the principal axes (cardinal red), the model is $\PT$-symmetric and hence below the critical values $\eps=L_{\lambda_1,\lambda_2,0,0}^\aubrydual/(2\pi)$ in the subcritical 
		and $\eta=L_{\lambda_1,\lambda_2,0,0}/(2\pi)$ in the supercritical phase 
		the spectrum is confined to the unit circle. 
		(b) Numerical evaluation for $(\lambda_1,\lambda_2)=(0.5,0.25)$, $(\lambda_1,\lambda_2)=(0.5,0.5)$ and $(\lambda_1,\lambda_2)=(0.25,0.5)$, respectively, with colors indicating the mean fractal dimension of the eigenstates.
	}
\end{figure*}
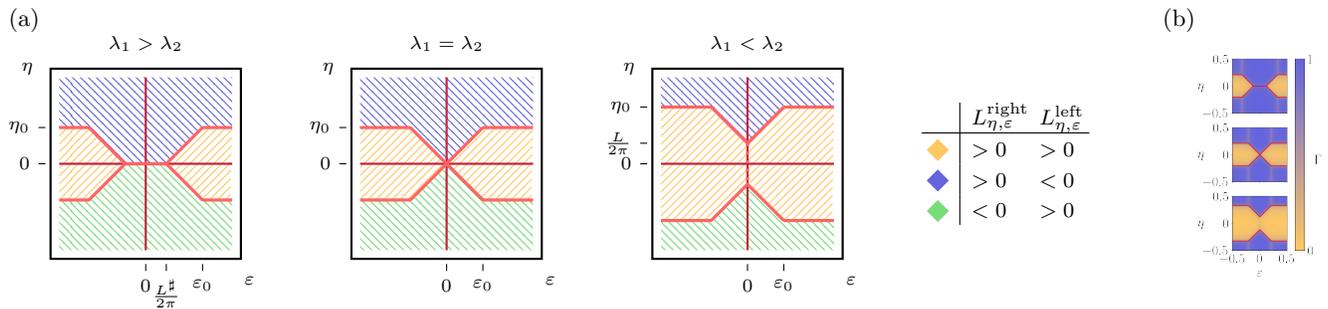

We will exploit heavily the duality between the non-reciprocal hopping and the complexified phase.
Since $W_{\lambda_1,\lambda_2,\eta,\varepsilon}$ and its transpose have the same spectrum, every statement about its spectrum implies the same statement for the spectrum of $W_{\lambda_2,\lambda_1,-\varepsilon,-\eta}$.

\section{Pseudo-unitarity and $\PT$-symmetry}
In the complexified regime, the PUAMO is not unitary anymore since $W_{\lambda_1,\lambda_2,\eta,\eps}^\adj\neq W_{\lambda_1,\lambda_2,\eta,\eps}^{-1}$. 
However, it is \emph{pseudo-unitary} \cite{mostafazadehPseudounitaryOperatorsPseudounitary2004}, whence the nomenclature: there exists a hermitian automorphism $\alpha$ such that $\alpha W_{\lambda_1,\lambda_2,\eta,\eps}^{-1}\alpha^{-1}=W_{\lambda_1,\lambda_2,\eta,\eps}^\adj$. In the current setting, this pseudo-unitarity is hidden: only after utilizing a gauge transformation \cite{CFO1, CFLOZ} we find that $W_{\lambda_1,\lambda_2,\eta,\eps}$ is pseudo-unitary for $\alpha=\Psym$, where $\Psym$ is the ``parity'' transformation that identifies the lattice points $n$ and $-n$ and acts locally as $\Psym_n=\sigma_2$, that is, $(\Psym\psi)_n=\sigma_2\psi_{-n}$; see Appendix ~\ref{sec:pt_sym} and Lemma \ref{lem:pseudo-unitarity} in particular. The pseudo-unitarity of $W_{\lambda_1,\lambda_2,\eta,\eps}$ has severe consequences: it implies that eigenvalues either lie on the unit circle, i.e., $|z|=1$, or come in pairs $(z,1/\overline z)$, see Corollary \ref{cor:spectrum_comes_in_pairs}.

For $\eta=0$ we unveil another hidden symmetry of $W_{\lambda_1,\lambda_2,0,\eps}$: the ``timeframe'' \cite{F2W,asbothBulkBoundaryCorrespondence2013}  $\widetilde W_{\lambda_1,\lambda_2,0,\eps}=Q_{\lambda_2,\eps}^{1/2}S_{\lambda_1} Q_{\lambda_2,\eps}^{1/2}$ satisfies 
\begin{equation*}
	(\mathcal{PT})\widetilde W_{\lambda_1,\lambda_2,0,\eps}(\mathcal{PT})^{-1}=\widetilde W_{\lambda_1,\lambda_2,0,\eps}^{-1},
\end{equation*}
where $\PT$ is the combined inversion of space and time that acts locally as $(\PT)_n=\sigma_3$, that is $(\PT\psi)_n=\sigma_3\overline{\psi_{-n}}$; see Lemma \ref{lem:PT_symmetry}. 
Since shifting the origin of time produces the same dynamics, this is an appropriate notion of $\PT$-symmetry that is flexible enough to accommodate our framework \cite{mochizukiExplicitDefinitionPT2016}. Note, however, that as soon as one admits symmetry operators whose local action depends on position, this change of timeframe becomes obsolete.

Employing Aubry duality, we conclude that the PUAMO is also $\PT$-symmetric whenever $\eps=0$.
However, when $\eta$ and $\eps$ are both non-zero, $\PT$-symmetry is broken since $(\PT)\widetilde W_{\lambda_1,\lambda_2,\eta,\eps}(\PT)^{-1}=\widetilde W_{\lambda_1,\lambda_2,-\eta,\eps}^{-1}$. 

\section{Lyapunov Exponents and sharp phase boundaries}
The PUAMO is a finite-difference operator, which allows us to 
recast the eigenvalue equation $W_{\lambda_1, \lambda_2, \eta, \eps}\psi = z\psi$ as

\begin{equation*}
	\begin{bmatrix}\psi_{n+1}^+\\\psi_{n}^-\end{bmatrix}
	= T_{n,z}\begin{bmatrix}\psi_{n}^+\\\psi_{n-1}^-\end{bmatrix},
\end{equation*}
where the \emph{transfer matrices} $T_{n,z}$ are given by (see Appendix \ref{sec:tmsAndLE})
\begin{widetext}
	\begin{equation*}
		\resizebox{\columnwidth-.5cm}{!}{$%
			\displaystyle{T_{n,z}=\frac{e^{2\pi\eta}}{\lambda_2\cos(2\pi(\Phi n+\vartheta))-i\lambda_2'}\begin{bmatrix}\lambda_1^{-1}z^{-1} +2\lambda_1'\lambda_1^{-1}\lambda_2\sin(2\pi(\Phi n+\vartheta)) + z{\lambda_1'}^2\lambda_1^{-1}&-e^{2\pi\eta}(\lambda_2\sin(2\pi(\Phi n+\vartheta))+\lambda_1' z)\\-e^{-2\pi\eta}(\lambda_2\sin(2\pi(\Phi n+\vartheta)) + \lambda_1' z)&\lambda_1 z\end{bmatrix}.}
			$}
	\end{equation*}
\end{widetext}

\def\lone{0.5}
\def\ltwo{0.25}
\def\etta{0.1}
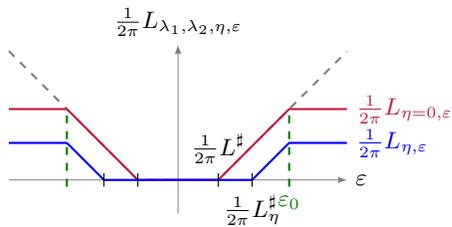
\begin{figure}[t]
	\begin{center}
		\def\xmax{.5}
\def\ticklength{.02}

\begin{tikzpicture}[scale=4.5]
	
	\draw[help lines,->] (-\xmax,0) -- (\xmax,0) node[black,right] {$\eps$};
	\draw[help lines,->] (0,-.1) -- (0,\xmax-.1) node[black,above] {$\frac1{2\pi}L_{\lambda_1,\lambda_2,\eta,\eps}$};
	
	\draw[domain=-\xmax:\xmax, samples=200, black!50!white, dashed, thick] plot (\x,{(max(0,l0(\lone,\ltwo)+2*pi*abs(\x)))/2/pi});
	
	\draw[domain=-\xmax:\xmax, samples=200, kardinal, thick] plot (\x,{(max(0,l0(\lone,\ltwo)+2*pi*abs(\x)-2*pi*max(0,abs(\x)-epszero(\ltwo)))/2/pi}) node[right] {$\frac1{2\pi}L_{\eta=0,\eps}$};
	
	\draw[domain=-\xmax:\xmax, samples=200, color=blue, thick] plot (\x,{min(max(0,-2*pi*\etta+l0(\lone,\ltwo)+2*pi*abs(\x)-2*pi*max(0,abs(\x)-epszero(\ltwo)))/2/pi,max(0,2*pi*\etta+l0(\lone,\ltwo)+2*pi*abs(\x)-2*pi*max(0,abs(\x)-epszero(\ltwo)))/2/pi)}) node[right] {$\frac1{2\pi}L_{\eta,\eps}$};
	
	
	\pgfmathparse{int(ceil(l0(\ltwo,\lone)/2/pi+\etta))} 
	\ifnum\pgfmathresult>0 %
	\draw	({l0(\ltwo,\lone)/2/pi+\etta},-\ticklength) node[below] {$\frac1{2\pi}L_\eta^\aubrydual$} -- ({l0(\ltwo,\lone)/2/pi+\etta},\ticklength)
	({-l0(\ltwo,\lone)/2/pi-\etta},-\ticklength) -- ({-l0(\ltwo,\lone)/2/pi-\etta},\ticklength);
	\else
	\fi
	
	\draw	({l0(\ltwo,\lone)/2/pi},-\ticklength) node[above=.2cm] {$\frac1{2\pi}L^\aubrydual$} -- ({l0(\ltwo,\lone)/2/pi},\ticklength)
	({-l0(\ltwo,\lone)/2/pi},-\ticklength) -- ({-l0(\ltwo,\lone)/2/pi},\ticklength);
	
	\draw[line width=.7, green, thick, dashed] ({epszero(\ltwo)},-\ticklength) node[below] {$\eps_0$} -- ({epszero(\ltwo)},{(l0(\lone,\ltwo)+2*pi*epszero(\ltwo))/2/pi})
	({-epszero(\ltwo)},-\ticklength) -- ({-epszero(\ltwo)},{(l0(\lone,\ltwo)+2*pi*epszero(\ltwo))/2/pi});
	
\end{tikzpicture}
	\end{center}
	\caption{\label{fig:lyap_exp_UAMO}
		Complexified Lyapunov exponent $L_{\lambda_1,\lambda_2,\eta,\eps}$ as a function of $\eps$ in the subcritical regime $(\lambda_1,\lambda_2)=(\lone,\ltwo)$ in the reciprocal setting $\eta=0$ (cardinal red) and for $\eta=\etta0$ (blue). It has two turning points at $L^\sharp/(2\pi)=L_{\lambda_1,\lambda_2,\eta,0}^\sharp/(2\pi)$ and at $\eps_0$. The gray dashed line symbolizes the Lyapunov exponent of the corresponding AMO in the reciprocal setting. It only has one turning point, and therefore no second phase transition.
	}
\end{figure}

The \emph{right Lyapunov exponent} $L^{\textrm{right}}(z) =L_{\lambda_1,\lambda_2,\eta,\eps}^{\textrm{right}}(z)$ measures the typical exponential growth (or decay) of eigenstates to the right:
\begin{equation*}
	L_{\lambda_1,\lambda_2,\eta,\eps}^{\textrm{right}}(z) = \lim_{n\to\infty} \frac{1}{n} \int_\bbT \log \|T_{n,z} \cdots T_{1,z}\| \, d\theta.
\end{equation*}
Similarly, we define $L_{\lambda_1,\lambda_2,\eta,\eps}^{\textrm{left}}(z)$ by iterating the transfer matrices to the left. We show in Appendix \ref{sec:tmsAndLE} that the non-reciprocality parameter $\eta$ contributes additively to the right and left Lyapunov exponents. Concretely, we show that
\begin{equation*}
	L_{\lambda_1,\lambda_2,\eta,\eps}^{\textrm{right}}(z)=2\pi\eta+L_{\lambda_1,\lambda_2,0,\eps}(z),
\end{equation*}
and that $L_{\lambda_1,\lambda_2,\eta,\eps}^{\textrm{left}}=L^{\textrm{right}}_{\lambda_1,\lambda_2,-\eta,\eps}$, see Proposition \ref{prop:Lyap_LR_eta}. In the reciprocal setting $\eta=0$, the left and right Lyapunov exponents therefore agree. We remark that this holds for any non-reciprocal split-step walk, that is, for any split-step walk $S_{\lambda,\eta}Q$ the non-reciprocality parameter contributes additively to the right and left Lyapunov exponents and the two are related by letting $\eta\mapsto-\eta$.

Moreover, for the reciprocal PUAMO, the Lyapunov exponent can be calculated explicitly \cite{CFO1}, see \eqref{eq:lyap_res}. Together with the contribution of $\eta$, we then obtain
\begin{equation*}
	\begin{split}
		L_{\lambda_1,\lambda_2,\eta,\eps}^{\textrm{right}}(z)
			=2\pi\eta +&\max\{0,  \log\lambda_0 
			\\
			& 
			+2\pi(|\eps|  -\max\{|\eps|-\eps_0,0\})\},
	\end{split}
\end{equation*}
where $\eps_0=\log\left[(1+\lambda_2')/\lambda_2\right]/(2\pi)$ marks the value of $\eps$ for which $T_{n,z}$ is not analytic in $\theta$, see also Section \ref{sec:second_phase_transition}. The overall inverse localization length is then
\begin{equation*}
	L_{\lambda_1,\lambda_2,\eta,\eps}(z) = \max\{0 , \min\{L_{\lambda_1,\lambda_2,\eta,\eps}^{\rm left}, L_{\lambda_1,\lambda_2,\eta,\eps}^{\rm right}\}\},
\end{equation*}
with dual $L_{\lambda_1,\lambda_2,\eta,\eps}^\aubrydual(z)=L_{\lambda_2,\lambda_1,\eps,\eta}(z)$.

The left and right Lyapunov exponents distinguish the phases in Figure \ref{fig:eta_eps_phasediagram}:
When positive, the eigenstates decay exponentially in the respective direction and the system localizes if both $L^{\textrm{left}},L^{\textrm{right}}>0$. When $\min\{L^{\textrm{left}}, L^{\textrm{right}}\}\leq0$, the eigenstates cease to decay and the system delocalizes.

\section{Constancy of the spectrum}
Transfer matrix techniques give access to ``global theory'' \cite{Avila2015Acta}: a major result of this theory is that, in the region of analyticity, the Lyapunov exponent is a continuous, convex and piecewise linear function of $\eps$ with slopes quantized by the integers.
A careful analysis of this function reveals subtle properties of the spectrum; see Appendix \ref{app:global_theory}.

\begin{figure*}[t]
	\includegraphics[width=.98\textwidth]{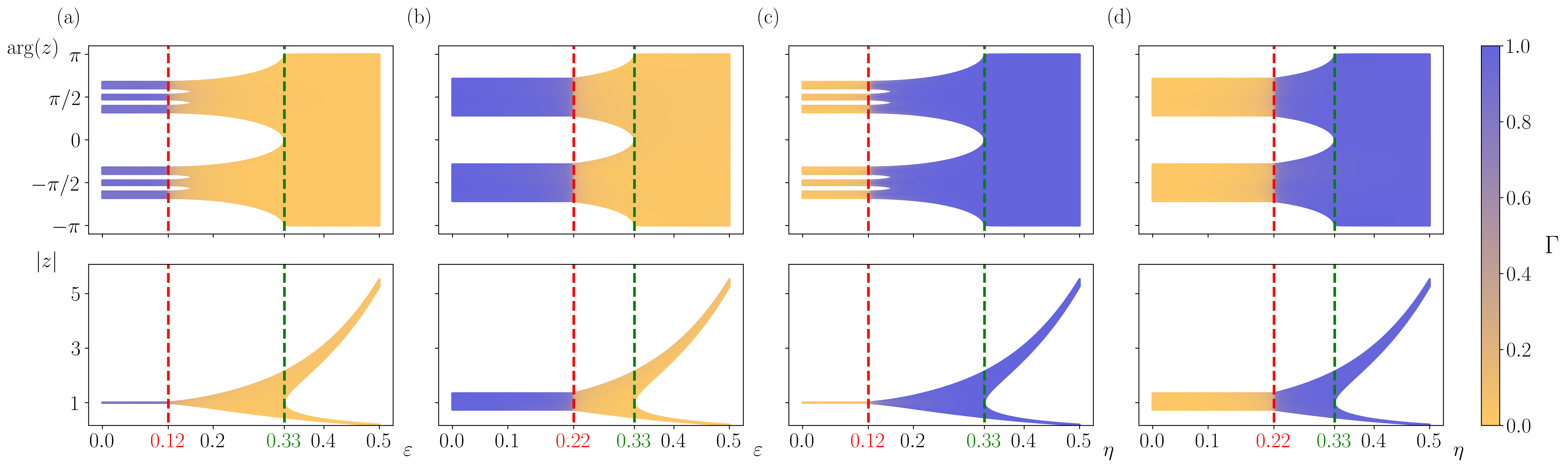}
	\caption{\label{fig:spectra_reci}The spectrum of the walk $W_{\lambda_1,\lambda_2,\eta,\eps}$ on a ring of size $N=610$ for $\Phi=(\sqrt{5}-1)/2$ and $\theta=0$ as a function of $\eps$ for $\eta=0,0.1/(2\pi)$ in the subcritical regime $0.5=\lambda_1>\lambda_2=0.25$ ((a) and (b)) and of $\eta$ for $\eps=0,0.1/(2\pi)$ in the supercritical regime $0.25=\lambda_1<\lambda_2=0.5$ ((c) and (d)). The color codifies the \emph{fractal dimension} $\Gamma$ which quantifies how extended or localized each eigenstate is \cite{thoulessElectronsDisorderedSystems1974}. As a rule of thumb, the more localized an eigenstate, the smaller its fractal dimension. The mobility edge (red) at $L_{\lambda_1,\lambda_2}^\aubrydual/(2\pi),L_{\lambda_1,\lambda_2}/(2\pi)\approx 0.1188$ and the spectral transition (green) at $\eta_0,\eps_0\approx 0.3284$ are clearly identifiable.
		Evidently, the subcritical case behaves in $\eps$ exactly as the supercritical case in $\eta$ (compare (a) with (c) and (b) with (d)), yet with opposite dynamical phase.
	}
\end{figure*}

For instance, in the subcritical case $\lambda_1>\lambda_2$, the graph of $L_{\lambda_1,\lambda_2,\eta =0,\eps}$ turns exactly at $\eps = \pm L_{\lambda_1,\lambda_2,\eta,0}^\aubrydual/(2\pi)$, see Figure \ref{fig:lyap_exp_UAMO}. From this, we infer that for $\eta=0$ the spectrum of $W_{\lambda_1,\lambda_2,0,\eps}$ is independent of $\eps$ in the regime $0 \le \eps\leq L_{\lambda_1,\lambda_2,0,0}^\aubrydual/(2\pi)$; see Figure~\ref{fig:spectra_reci} (a). Indeed, the complement of the spectrum consists of the $z$ for which $L_{\lambda_1,\lambda_2,\eta=0,\eps}(z)$ is positive and does not change slope at $\eps$ \cite{Avila2015Acta}, and this is the case exactly through the indicated $\eps$-region. This persists for fixed $\eta \neq 0$: the spectrum of $W_{\lambda_1,\lambda_2,\eta,\eps}$ is constant in the region 
\begin{equation*}
	0 \le |\eps| \le L_{\lambda_1,\lambda_2, \eta,0}^\aubrydual/(2\pi)
= L_{\lambda_1, \lambda_2,0,0}^\aubrydual/(2\pi) + |\eta|;
\end{equation*}
see Figure~\ref{fig:spectra_reci} (b). A similar discussion unveils constancy of the spectrum of the \emph{supercritical} model in the regime $0 \leq \eta \leq L_{\lambda_1, \lambda_2,0,\eps}/(2\pi)$; see Figure~\ref{fig:spectra_reci} (c) and (d).

\section{First Phase Transition}
Once $\eps$ increases through the threshold value $L_{\lambda_1,\lambda_2,\eta,0}^\aubrydual/(2\pi)$ where the slope of $L_{\lambda_1,\lambda_2,\eta,\eps}$ changes \cite{geMultiplicativeJensenFormula2023,CFO1}, states transition from  extended to localized. 
Similarly, increasing $\eta$ through the transition point $L_{\lambda_1,\lambda_2,0,\eps}/(2\pi)$ produces a transition from localization to delocalization in which the non-reciprocal hopping overpowers the localization effects of the coin.
These phase transitions may be absent: for example, in the subcritical case $\lambda_1>\lambda_2$ for $0\leq|\eps|\leq L_{\lambda_1, \lambda_2,0,\eps}^\aubrydual/(2\pi)$, delocalization holds since $L_{\lambda_1,\lambda_2,\eta,\eps} = 0$ for all $\eta$.

The delocalization in $\eta$ can be understood intuitively: we show in Appendix \ref{app:skin} that the PUAMO with $\eta\neq0$ is similar to that with $\eta=0$ via a ``skin transformation''. Therefore, if $\psi_n$ is an eigenstate of the reciprocal model, $\psi_n^{(\eta)}=e^{-2\pi\eta n}\psi_n$ is an eigenstate of the non-reciprocal model. If $\psi_n$ is localized with inverse localization length $\ell=L_{\lambda,\eta=0,\eps}(z)>0$, $\psi_n^{(\eta)}$ is localized and decays with inverse localization length $\ell-2\pi|\eta|$, as long as $2\pi|\eta|< \ell$.
As soon as $2\pi|\eta|\geq \ell$, $\psi_n^{(\eta)}$ delocalizes to one side  (while still decaying to the other). This is precisely what is measured by $L^{\mathrm{left}}$ and $L^{\mathrm{right}}$; see also Figure \ref{fig:eta_eps_phasediagram}.

Along the principal axes in Figure \ref{fig:eta_eps_phasediagram}, the $\PT$-symmetry is ``spontaneously broken'' at these critical values, that is, some of the spectrum moves off of the unit circle, see Figures \ref{fig:spectra_reci} and \ref{fig:UC}. In this case, these critical values demark a topological phase transition, which is measured by an abrupt jump of a winding number as discussed below.

\section{Second Phase Transition}\label{sec:second_phase_transition}
The transfer matrices $T_{n,z}$ are analytic in $\theta$ until $\eps$ reaches the value
\begin{equation*}
	\eps_0 = \frac{1}{2\pi} \sinh^{-1} \left[\frac{\lambda_2'}{\lambda_2}\right]=\frac{1}{2\pi}\log\left[\frac{1+\lambda_2'}{\lambda_2}\right].
\end{equation*}
At this value the walk partially decouples, that is, the diagonal entries of the coins $Q_{\lambda_2,\eps}$ come arbitrarily close to $0$ infinitely often. 

This is the key feature producing a second transition: Indeed, for $\eps$ with $L^\aubrydual_{\lambda_1,\lambda_2,\eta,0}/(2\pi) < \eps < \eps_0$, the spectrum contains some parts on and some parts off the unit circle. At $\eps=\eps_0$ these components merge and for $\eps > \eps_0$, the spectrum is completely off the unit circle, see Figures \ref{fig:spectra_reci} and \ref{fig:UC}. Employing Aubry duality we conclude similarly that the walk has some spectrum on and off the unit circle for $L_{\lambda_1,\lambda_2,0,\eps}/(2\pi)<\eta<\eta_0=\sinh^{-1} (\lambda_1'/\lambda_1)/(2\pi)$ and that all eigenvalues leave once $\eta>\eta_0$. 
We explain how to derive this using transfer matrix methods in Appendix~\ref{sec:tmsAndLE}. 
We emphasize that this does not occur in the Hamiltonian setting, where the transfer matrices are always analytic.
A direct consequence of the second transition is that it breaks duality in the sense that there are $(\eta,\eps)$ that do not change phase when dualizing: for fixed $\lambda_1,\lambda_2$, both the PUAMO and its dual exhibit transport whenever $\eta,\eps\geq\max\{\eta_0,\eps_0\}$.

\begin{figure}[t]
	\begin{center}
		\includegraphics[width=.45\textwidth]{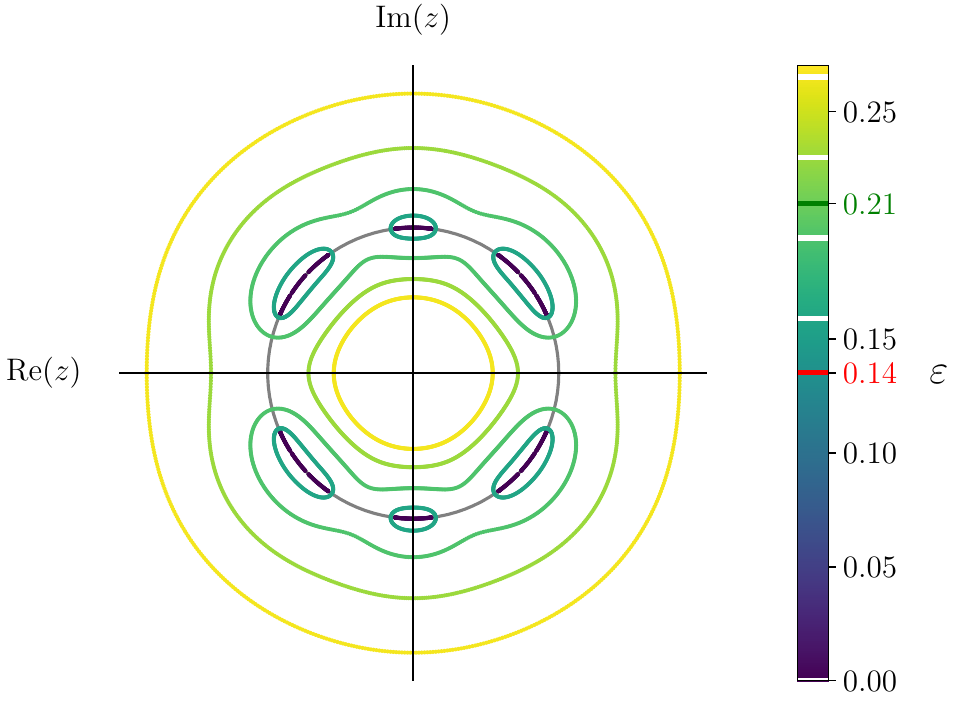}
	\end{center}
	\caption{\label{fig:UC}
		The spectra of $W_{\lambda_1,\lambda_2,\eta=0,\eps}$ for $(\lambda_1,\lambda_2)=(0.9,0.5)$ with $\eps$ values indicated by white markers. For $\eps>L^\aubrydual/(2\pi)\approx0.14$ the eigenvalues form ``bubbles'' around the bands of the unitary setting yielding a non-trivial winding number $\wind_\eps(z)$. At $\eps_0\approx0.21$ these bubbles merge to two ``circles'' within and without the unit circle.
	}
\end{figure}

\section{Winding numbers in the $\PT$-symmetric setting}
As discussed above, the $\PT$-symmetry of the PUAMO along the principal axes in Figure \ref{fig:eta_eps_phasediagram} is spontaneously broken at the first phase transition, and some spectrum moves off of the circle, see Figure \ref{fig:UC}. This transition is topological and measured by a spectral winding number.

To define this invariant, fix $\eta=0$ and let $z\in \partial \bbD$ lie in a gap of the spectrum of the UAMO which is a Cantor set \cite{CFO1,UAMO_Dry_Martinis}. Consider the continued fraction expansion of $\Phi$ with convergents $\Phi_k=\nr_k/\dr_k$. The corresponding PUAMO is $N:=\dr_k$-periodic, and we denote the resulting matrix with periodic boundary conditions by $W_N=W_N(\theta+i\eps)$.
We show in Appendix \ref{app:winding} that the winding number of $\theta\mapsto\sigma(W(\theta+i\eps))$ around $z$ defined by \cite{gongTopologicalPhasesNonHermitian2018,longhiTopologicalPhaseTransition2019}
\begin{equation*}\label{eq:wind_def}
	\wind_\eps(z)=\lim_{N\to\infty}\frac1N\frac1{2\pi i}\int_0^1 d\theta\:\partial_\theta\log\det(W_N(\theta+i\eps)-z\idty)
\end{equation*}
is quantized and takes only three values: if $|\eps|<L^\aubrydual/(2\pi)$ then $\wind_\eps(z)=0$ independent of $z$, if $L^\aubrydual(z)/(2\pi)\leq|\eps|\leq\eps_0$ then there exist $z\in\partial\bbD$ such that $|\wind_\eps(z)|=1$, and if $|\eps|\geq\eps_0$ then $|\wind_\eps(z)|=1$ for all $z\in\partial\bbD$. 

As long as $|\eps|<\eps_0$ the winding number agrees with the ``acceleration'' of $L_\eps(z)$, i.e., $\wind_\eps(z)=-\partial_\eps L_\eps(z)/(2\pi)$ \cite{geMultiplicativeJensenFormula2023}. However, this correspondence fails for $|\eps|\geq\eps_0$: $|\nu_\eps(z)|=1$ is insensitive to the second phase transition whereas the acceleration vanishes, see Figure \ref{fig:lyap_exp_UAMO}. Via Aubry duality, the same applies along the other principle axis in Figure \ref{fig:eta_eps_phasediagram} with the roles of $\eta$ and $\eps$ interchanged.
Similar transitions characterized by the winding number of an operator (or of a suitable effective Hamiltonian) have been observed in other contexts \cite{weidemannTopologicalTriplePhase2022, liuLocalizationTransitionSpectrum2021a}.

\section{Conclusion and Outlook}
We introduced a discrete-time simulator of non-hermitian extensions of the AMO and characterized its dynamics in the regime of non-reciprocal hopping and complex phase. We showed that this quantum walk possesses two phase transitions: one that signals a mobility edge between a metallic and an insulating phase and a second one that caps off the critical values for the mobility edge. This second transition is exclusive to the discrete-time setting. In the $\PT$-symmetric parameter ranges, the first phase transition is topological and indicates when the spectrum leaves the unit circle. Our methods are largely independent of the concrete model, wherefore we expect them to apply to other models of interest such as \cite{weidemannTopologicalTriplePhase2022}.

The PUAMO is within reach of current experimental methods: For example, the time-multiplexed photonic quantum walk setups of \cite{linTopologicalPhaseTransitions2022,weidemannTopologicalTriplePhase2022} promise to be directly adaptable for non-vanishing $\eta$ and $\eps=0$ \cite{LinXuePC}. 
To probe the dependence on $\eps$, a promising approach is to dualize: then $\eps$ plays the role of a non-reciprocal hopping, which is admissible to experimental manipulation. When this dual model is supercritical and therefore localized, from our analysis one expects a mobility edge at $\eps=L_{\lambda_1,\lambda_2}/(2\pi)$.

\paragraph{Note added in proof:} While this paper was under review, the PUAMO was experimentally implemented in a photonic setup using single photons \cite{ExpUAMO}. Already after a small number of timesteps, the different regimes are clearly distinguishable in the experimental data in the unitary as well as in the pseudo-unitary regime. Quantifying transport by the second moment of the time-evolved position operator, the whole phase diagram of Figure \ref{fig:phase_diagram} was experimentally mapped out and confirms the theoretical phase transition in the $(\lambda_1,\lambda_2)$-plane. Moreover, the first and second phase transitions were successfully probed by tuning the non-reciprocality parameter $\eta$ and exploiting Aubry duality.

\section*{Acknowledgements}
C.~Cedzich and J.~Fillman thank S. Longhi, A. Mostafazadeh and Q. Zhou for insightful discussions and explanations on $\PT$-symmetric systems.
C.~Cedzich thanks L. Bittel for helpful suggestions for the numerics. 
C.~Cedzich\ was supported in part by the Deutsche Forschungsgemeinschaft (DFG, German Research Foundation) under the grant number 441423094.
J.~Fillman\ was supported in part by National Science Foundation   grants DMS-2213196 and DMS-2513006
and Simons Foundation Grant MPS-TSM-00013720.
J.~Fillman also thanks the American Institute of Mathematics for hospitality during a recent SQuaRE program.

\bibliographystyle{quantum_mod}

\bibliography{PUAMO-quantum_03_final-bib}

\onecolumn\newpage
\setcounter{secnumdepth}{2}

\numberwithin{theorem}{section}
\numberwithin{equation}{section}

\appendix

\section{Pseudo-unitarity and $\PT$-symmetry}\label{sec:pt_sym}

Whenever $\eta\neq0$ or $\eps\neq0$, the PUAMO $W_{\lambda_1,\lambda_2,\eta,\eps}$ is not unitary anymore: it no longer satisfies $W^{-1}=W^\adj$. However, it satisfies a generalized unitarity condition: let $\alpha$ be a linear, bounded, invertible, hermitian operator. Then an operator is called \emph{$\alpha$-pseudo-unitary}, if it is unitary ``up to $\alpha$'', i.e., if it satisfies
\begin{equation*}
	\alpha W^{-1}\alpha^{-1}=W^\adj.
\end{equation*}
In this case, $\alpha$ defines a possibly indefinite inner product $\langle\cdot,\cdot\rangle_\alpha:=\langle\cdot,\alpha\ \cdot\rangle$ that is invariant under $W$: with the above definition one easily verifies that an $\alpha$-pseudo-unitary operator satisfies
\begin{equation*}
	\langle W\psi,W\phi\rangle_\alpha=\langle \psi,\phi\rangle_\alpha.
\end{equation*}

In the following, we call any involutive unitary operator $\Psym$ on a Hilbert space with a lattice structure a \emph{parity transformation} if $\Psym$ identifies the lattice points $n$ and $-n$. A parity transformation thus acts on $\psi\in\ell^2(\bbZ)\otimes\bbC^2$ as 
\begin{equation}\label{eq:PsymAction}
	[\Psym\psi]_n
	=\Psym_n\psi_{-n},
\end{equation}
where $\Psym_n$ is a finite-dimensional unitary that implements the local action of $\Psym$. From $\Psym^2=\idty$ and the unitarity of $\Psym$ we find that $\Psym_{-n}=\Psym_{n}^{-1}=\Psym_{n}^\adj$. 
One can also shift the center of the parity transformation by considering instead $\psi_n\mapsto \Psym_n\psi_{-n+q}$, but for us it suffices to take $q=0$.
We call an operator $W$ ``parity symmetric'' if there is a parity transformation $\Psym$ such that $\Psym W\Psym^{-1}=W$.
\begin{example}
	The parity transformation on $\ell^2(\bbZ)\otimes\bbC^2$ characterized by $\Psym_n\equiv\sigma_1$ acts as
	\begin{equation*}
		[\Psym\psi]_n^\pm=\psi_{-n}^\mp.
	\end{equation*}
	One can show that for $\theta=1/4$ the (realified) UAMO $W_{\lambda_1,\lambda_2}$ is parity symmetric with respect to this parity transform, see \cite{CFLOZ} where this $\calP$ is called a ``reflection''.
\end{example}
With this at hand we can identify the $\alpha$ with respect to which the PUAMO is pseudo-unitary:

\begin{lemma}\label{lem:pseudo-unitarity}
	The PUAMO $W_{\lambda_1,\lambda_2,\eta,\eps}$ is $\Psym$-pseudo-unitary, where $\Psym$ is the parity transform that acts locally as $\Psym_n=\sigma_2$.
\end{lemma}
\begin{proof}
	Fixing $\theta = 1/4$ and using the (unitary) gauge transformation $D=\bigoplus_{n\in\bbZ}D_n$ from \cite{CFO1, CFLOZ}, we can consider the walk $W^\bbR = S_\lambda Q^\bbR$ with coins where we replace the diagonal entries with their moduli and leave the other entries unchanged, that is,
	\begin{equation}\label{eq:coin_realified}
		Q^\bbR_n = \begin{bmatrix} (1-\lambda_2^2\cos(2\pi(n\Phi+i\eps))^2)^{1/2}
			& - \lambda_2 \cos(2\pi(n\Phi+i\eps))  \\ 
			\lambda_2 \cos(2\pi(n\Phi+i\eps)) & (1-\lambda_2^2\cos(2\pi(n\Phi +i\eps))^2)^{1/2}\end{bmatrix}.
	\end{equation}				
	Denote by $\Psym$ the parity transformation that acts locally as $\Psym_n=\sigma_2$. Then, by direct computation one finds that $\Psym(Q^\bbR)^{-1}\Psym^{-1}=(Q^\bbR)^\adj$. Moreover, since $\Psym$ reverses the direction of the lattice shift, one verifies by direct computation that $\Psym S_{\lambda,\eta}^{-1}\Psym=S_{\lambda,\eta}^\adj$, which proves the lemma.				
\end{proof}
\begin{remark}
	The inner product induced by $\Psym$ is not positive definite: for any $n$ one has for $\psi_n:=(\delta_{-n}\otimes[i, 0]^\top+\delta_n\otimes[0, 1]^\top)/\sqrt2$ that
	\begin{align*}
		\langle\psi_n,\psi_n\rangle_\Psym 	&=\langle\psi_n,\Psym\psi_n\rangle 
		=-1.
	\end{align*}
\end{remark}
If $W_{\lambda_1,\lambda_2,\eta,\eps}$ would be pseudo-unitary with respect to a positive definite inner product throughout the whole parameter space, its spectrum would be confined to the unit circle.
Since this is not the case, one merely has that
\begin{coro}\label{cor:spectrum_comes_in_pairs}
	The spectrum of $W_{\lambda_1,\lambda_2,\eta,\eps}$ is either on the unit circle {\rm(}$|z|=1${\rm)} or comes in $(z,1/\overline{z})$ pairs.
\end{coro}
\begin{proof}
	Let $z\in\sigma(W_{\lambda_1,\lambda_2,\eta,\eps})=\sigma(W_{\lambda_1,\lambda_2,\eta,\eps}^\realified)$. Then $\overline z\in\sigma(W_{\lambda_1,\lambda_2,\eta,\eps}^\adj)$ so by the above lemma $\overline z\in\sigma(W_{\lambda_1,\lambda_2,\eta,\eps}^{-1})$ so $1/\overline z\in\sigma(W_{\lambda_1,\lambda_2,\eta,\eps})$.
\end{proof}

In the reciprocal case $\eta=0$, $W_{\lambda_1,\lambda_2,0,\eps}$ is additionally $\PT$-symmetric. Here, $\mathcal P$ refers to a unitary parity transformation as defined above, while the $\mathcal T$ is an anti-unitary ``time-reversal'' transformation that, as the name suggests, inverts the direction of time. 
It acts on $\psi\in\ell^2(\bbZ)\otimes\bbC^2$ as
\begin{equation}\label{eq:TsymAction}
	[\Tsym\psi]_n
	=\Tsym_n\mathcal K\psi_{n}
	=\Tsym_n\overline{\psi_{n}},
\end{equation}
where $\mathcal K$ denotes complex conjugation in the standard basis and $\Tsym^2=\pm\idty$. We call an operator $W$ ``time-reversal symmetric'' if there exists a time-reversal transformation $\Tsym$ such that $\Tsym W\Tsym^{-1}=W^{-1}$. The intuition behind this is that $W^{-1}$ implements the same dynamics as $W$ but in reversed time direction. 

Similarly, we call an anti-unitary and involutive transformation $\PT$ a ``parity-time transformation'' if it acts as a combined inversion of space and time, that is,
\begin{equation}\label{eq:PT_action}
	(\PT\psi)_n=(\PT)_n\overline{\psi_{-n}},
\end{equation}
where $(\PT)_n$ is the local action of the transformation. Moreover, we call a time-evolution operator $W$ ``parity-time-'' or ``$\PT$-symmetric'' if there exists a parity-time transformation $\PT$ such that
\begin{equation}\label{eq:pt_sym_def}
	(\PT)W(\PT)^{-1}=W^{-1}.
\end{equation}
Note that this definition is consistent with the usual definition of $\PT$-symmetric Hamiltonians, for which one demands that $(\PT)H(\PT)^{-1}=H$: if $W=\exp[iH]$, then $H$ is $\PT$-symmetric if and only if $W$ is.

\begin{remark}
	\mbox{}\\
	\begin{enumerate}
		\item	Of course, if $W$ is symmetric for some $X\in\{\Psym,\Tsym,\PT\}$, then $UWU^{-1}$ is symmetric for $XU^{-1}$ for any unitary $U$.
		\item	It is sometimes advantageous to shift the origin of time, or ``choose a different timeframe'' \cite{asbothBulkBoundaryCorrespondence2013,F2W}: This leads to the same dynamics (with different initial state), but can make it substantially easier to identify a symmetry.
				For example, in a situation where $W=AB$ and $A$ depends on position while $B$ does not, one cannot find a position-independent time-reversal operator since taking the inverse in the definition of time-reversal symmetry swaps the order of $A$ and $B$. To find a position-independent time-reversal symmetry one might instead consider the timeframe $\widetilde W=A^{\frac12}BA^{\frac12}$.
	\end{enumerate}
\end{remark}

\begin{lemma}\label{lem:PT_symmetry}
	For $\eta=0$ and $\theta=1/4$, the ``timeframe'' $\widetilde W_{\lambda_1,\lambda_2,0,\eps}=Q_{\lambda_2,\eps}^{1/2}S_{\lambda_1} Q_{\lambda_2,\eps}^{1/2}$ satisfies 
	\begin{equation*}
		(\mathcal{PT})\widetilde W_{\lambda_1,\lambda_2,0,\eps}(\mathcal{PT})^{-1}=\widetilde W_{\lambda_1,\lambda_2,0,\eps}^{-1},
	\end{equation*}
	where $\PT$ is the parity-time transformation that acts locally as $(\PT)_n=\sigma_3$.
\end{lemma}
\begin{proof}
	Fixing $\eta=0$, $\theta = 1/4$ and using the (unitary) gauge transformation $D=\bigoplus_{n\in\bbZ}D_n$ from \cite{CFO1, CFLOZ}, we can consider the walk $W^\bbR = S_\lambda Q^\bbR$ with coins as in \eqref{eq:coin_realified}. 
	By a direct computation, since $S_\lambda$ is real and a parity-transform $n\mapsto-n$ reverses the direction of the lattice shift $n\mapsto n+1$ on $\ell^2(\bbZ)$ one verifies from the definition of $S_\lambda$ and \eqref{eq:PT_action} that
	\begin{equation}
		(\PT)S_{\lambda}(\PT)^{-1}=\left[\bigoplus_n(\PT)_n\sigma_3 \right]S_{\lambda}^{-1}\left[\bigoplus_n\sigma_3(\PT)_n^{-1}\right].
	\end{equation}
	Moreover, from the expression of the realified coin given in \eqref{eq:coin_realified} one has that
	\begin{equation}\label{eq:Q-ncc}
		\overline{Q_{-n}^\bbR}= \sigma_3 (Q_n^\bbR)^{-1}\sigma_3,
	\end{equation}
	and thus we obtain $(\PT)\widetilde W (\PT)^{-1}$ for the timeframe $\widetilde W = (Q^\bbR)^{\frac12}S_\lambda(Q^\bbR)^{\frac12}$ for $(\PT)_n = \sigma_3$. 
\end{proof}

Note that once we incorporate a non-reciprocality parameter $\eta>0$, the PUAMO is not $\PT$-symmetric anymore. Instead we have that
\begin{equation*}
	(\PT)\widetilde W_{\eta}(\PT)^{-1}=\widetilde W_{-\eta}^{-1}
\end{equation*}
which can be seen either by direct computation, or by first pulling out $\eta$ using the skin transformation of Appendix \ref{app:skin}, commuting it through $\PT$, applying $\PT$-symmetry for the reciprocal model and finally putting the non-reciprocality back in.

\begin{remark}\mbox{}\\
	\begin{enumerate}
		\item Proving $\PT$-symmetry for one particular $\theta$ is sufficient, since for $\Phi$ irrational the spectrum is independent of $\theta$ by minimality of $\theta\mapsto\theta+\Phi$ \cite{CFO1}.
		\item If one allows for symmetry operators that depend on position, one can prove $\PT$-symmetry directly for $W_{\lambda_1,\lambda_2,0,\eps}^\realified$ with local symmetry $(\PT)_n=(Q_n^\realified)^{-1}\sigma_3$.
		\item In the timeframe $\widetilde W = (Q^\bbR)^{\frac12}S_\lambda(Q^\bbR)^{\frac12}$ it is straightforward to see that $W_{\lambda_1,\lambda_2,\Phi,\theta}$ is also chiral symmetric, i.e., there is a local unitary $\Gamma=\bigoplus_\bbZ\Gamma_n$ such that $\Gamma\widetilde W\Gamma^{-1}=\widetilde W^{-1}$ for $\Gamma_n=\sigma_1$. See also \cite{WeAreSchur}. This also extends to the non-reciprocal case. Therefore, the spectrum of $W$ comes in pairs $(z,z^{-1})$.
		\item The action of complex conjugation $\mathcal K$ (in the standard basis) on $W^\realified$ amounts to $\Phi\mapsto-\Phi$. Since the spectra of $W^\realified_\Phi$ and $W^\realified_{-\Phi}$ are the same by abstract $C^*$-algebra arguments \cite{rieffelAlgebrasAssociatedIrrational1981}, the spectrum of the PUAMO must come in pairs $(z,\overline z)$.
	\end{enumerate}
\end{remark}

\section{A short course in global theory of analytic one-frequency cocycles}
\label{app:global_theory}

In our analysis we rely heavily on Avila's global theory of analytic one-frequency cocycles \cite{Avila2015Acta}. For completeness, let us therefore first recall its main results and then apply them to the (P)UAMO. For further background and references about dynamical cocycles and their relation to operator theory, see \cite{ESO1, ESO2, MarxJito2017ETDS}.

Let $A$ be a continuous function from $\bbT$ to $\bbM(2,\bbC)$, the $2\times2$ complex matrices, and let $\Phi\in\bbT$ be irrational. Then, we call the skew-product
\begin{equation*}
	(\Phi,A):\bbT\times\bbC^2\to\bbT\times\bbC^2,\qquad (\theta,v)\mapsto(\theta+\Phi,A(\theta)v)
\end{equation*}
a \emph{quasiperiodic cocycle} with iterates $(\Phi,A)^n=(n\Phi,A^n)$ where $A^n(\theta)=\prod_{k=n-1}^0A(\theta+k\Phi)$. The \emph{Lyapunov exponent} of the quasiperiodic cocycle $(\Phi,A)$ is defined as
\begin{equation}
	L(\Phi,A)=\lim_{n\to\infty}\frac1{n}\int_0^1 \log\left\|A^n(\theta) \right\| d\theta.
\end{equation}

If $A$ is analytic and admits an analytic extension to the strip $\{\theta+i\eps: |\eps|<\delta\}$, we define the complexified cocycle $A_\eps(\theta):=A(\theta+i\eps)$ and the associated Lyapunov exponent by $L(\Phi,A,\eps):=L(\Phi,A_\eps)$.
An important observation of \cite{Avila2015Acta} is that the function $\eps\mapsto L(\Phi,A,\eps)$ is continuous, convex and piecewise affine. Moreover, its \emph{acceleration} $\omega$ defined by
\begin{equation}\label{eq:acceleration_def}
	\omega(\Phi,A,\widetilde\eps)=\lim_{\eps\to0^+}\frac1{2\pi\eps}(L(\Phi,A,\widetilde\eps+\eps)-L(\Phi,A,\widetilde\eps))
\end{equation}
is quantized for all $|\widetilde\eps|<\delta$ \cite{JitoMarx2012CMP,JitoMarx2012CMPErr,Avila2015Acta}, that is,
\begin{equation*}
	\omega(\Phi,A,\widetilde\eps)\in\frac12\bbZ.
\end{equation*}
Moreover, if $A:\bbT\to\SL(2,\bbC)$, one even has $\omega(\Phi,A,\widetilde\eps)\in\bbZ$. An important virtue of global theory is that it facilitates the explicit computation of the Lyapunov exponent which is difficult in general; see \cite{CFO1} for the UAMO and \cite{CFLOZ} for a related model.
\medskip

Using these properties, one can classify the behaviour of quasiperiodic cocycles: An analytic $\SL(2,\bbC)$-cocycle $(\Phi,A)$ is called \emph{uniformly hyperbolic} if for some constants $c,\lambda>0$ one has
\begin{equation*}
	\|A^n(\theta)\|\geq ce^{\lambda|n|}
\end{equation*}
uniformly in $n$ and $\theta\in\bbT$.
If an analytic $\SL(2,\bbC)$-cocycle  can be conjugated to real matrices, one has the symmetry $L(\Phi,A_{\eps}) = L(\Phi,A_{-\eps})$. In this setting, if the cocycle is not uniformly hyperbolic, it is said to be
\begin{enumerate}[itemsep=0pt,label=(\arabic*)]
	\item \emph{supercritical} if $L(\Phi,A)>0$.
	\item \emph{subcritical} if there exists $\tilde\eps>0$ such that $L(\Phi,A_\eps)=0$ for all $|\eps|<\tilde\eps$.
	\item \emph{critical} if $L(\Phi,A)=0$ and $(\Phi,A)$ is not subcritical.
\end{enumerate}
These different possibilities are illustrated in Figure \ref{fig:global_theory}. The super-/sub-/criticality can also be characterized via the acceleration in the following way:
\begin{center}
	\begin{tabular}{c|cc}
		& $\omega(\Phi,A,0)=0$ & $\omega(\Phi,A,0)>0$ \\\hline
		$L_{\eps=0}=0$ & subcritical & critical \\
		$L_{\eps=0}>0$ & unif. hyperbolic & supercritical
	\end{tabular}
\end{center}

\begin{figure}[t]
	\begin{center}
		\begin{tabular}{cccc}
		\def\lone{.5}
		\def\ltwo{.1}
		\def\col{magenta}
		\begin{tikzpicture}[scale=1.2]
	
	\draw[help lines,->] (-1,0) -- (1,0) node[black,right] {$\eps$};
	\draw[help lines,->] (0,-.2) -- (0,1.2) node[black,above] {$L(A_\eps)$};
	
	\draw[domain=-1:1, samples=200, color=\col, thick] plot (\x,{(max(0,l0(\lone,\ltwo)+2*pi*abs(\x))+2)/2/pi});
		
\end{tikzpicture}
		&
		\def\col{myorange}
		\def\lone{.1}
		\def\ltwo{.5}
		\begin{tikzpicture}[scale=1.2]
		
	\draw[help lines,->] (-1,0) -- (1,0) node[black,right] {$\eps$};
	\draw[help lines,->] (0,-.2) -- (0,1.2) node[black,above] {$L(A_\eps)$};
	
	\draw[domain=-1:1, samples=200, color=\col, thick] plot (\x,{(max(0,l0(\lone,\ltwo)+2*pi*abs(\x)))/2/pi});
		
\end{tikzpicture}
		&
		\def\col{myblue}
		\def\lone{.5}
		\def\ltwo{.1}
		\begin{tikzpicture}[scale=1.2]
		
	\draw[help lines,->] (-1,0) -- (1,0) node[black,right] {$\eps$};
	\draw[help lines,->] (0,-.2) -- (0,1.2) node[black,above] {$L(A_\eps)$};
	
	\draw[domain=-1:1, samples=200, color=\col, thick] plot (\x,{(max(0,l0(\lone,\ltwo)+2*pi*abs(\x)))/2/pi});
		
\end{tikzpicture}		
		&
		\def\col{myred}
		\def\lone{.5}
		\def\ltwo{.5}
		\begin{tikzpicture}[scale=1.2]
		
	\draw[help lines,->] (-1,0) -- (1,0) node[black,right] {$\eps$};
	\draw[help lines,->] (0,-.2) -- (0,1.2) node[black,above] {$L(A_\eps)$};
	
	\draw[domain=-1:1, samples=200, color=\col, thick] plot (\x,{(max(0,l0(\lone,\ltwo)+2*pi*abs(\x)))/2/pi});
		
\end{tikzpicture}
	\end{tabular}
	\end{center}
	\caption{\label{fig:global_theory}Illustrations of graphs of the Lyapunov exponents $L(A_\eps)$ in the possible regimes for the corresponding cocycle $A_\eps$ as a function of $\eps$. From left to right, $A$ is uniformly hyperbolic, supercritical, subcritical, and critical, respectively. The color coding is consistent with the one in the phase diagram in Fig. \ref{fig:phase_diagram}.}
\end{figure}
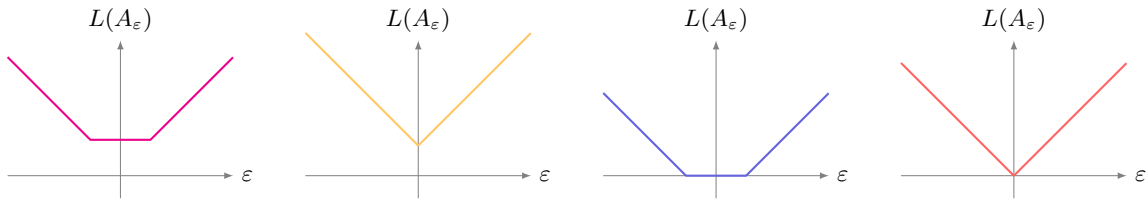

In the case where $A_\eps$ is the transfer matrix cocycle of an underlying quasiperiodic operator, this cocycle classification helps in determining the spectrum: $z\in\bbC$ is an element of the spectrum if and only if $A_\eps(z)$ is not uniformly hyperbolic \cite{DFLY2016DCDS, Johnson1986JDE}.
Moreover, as we explain in the next section, when applied to the transfer matrix cocycle of the UAMO $W_{\lambda_1,\lambda_2}$ the above characterization of cocycles provides the nomenclature in the phase diagram in Figure \ref{fig:phase_diagram} which is crucial in the determination of the spectral type of $W_{\lambda_1,\lambda_2}$ as well as its dynamical properties in \cite{CFO1}.
	
\section{Transfer Matrices and Lyapunov Exponents}\label{sec:tmsAndLE}

\subsection{Transfer Matrices and Lyapunov Exponents for non-reciprocal split-step walks}
An important tool for studying the behaviour of eigensolutions on the one-dimensional lattice are transfer matrix techniques. In the setting of split-step quantum walks with coupling constant $\lambda$ and non-reciprocality parameter $\eta$ this is captured by the following:
\begin{prop} \label{prop:cocycle:tmMain}
	Suppose $\lambda \in (0,1]$ and consider the split-step quantum walk $W = S_{\lambda,\eta} Q$ with coupling constant $\lambda$, non-reciprocality parameter $\eta$ and local coins $Q_n=(q_n^{ij})$, $i,j\in\{1,2\}$. 
	\begin{enumerate}[label={\rm (\alph*)}]
		\item If $z \in \bbC \setminus \{0\}$, $W\psi = z\psi$, and $Q_n$ is not an off-diagonal matrix, then
		\begin{equation} \label{eq:tm:main1}
			\begin{bmatrix}\psi_{n+1}^+\\\psi_{n}^-\end{bmatrix}
			= T_{n,z}(\eta)\begin{bmatrix}\psi_{n}^+\\\psi_{n-1}^-\end{bmatrix},
		\end{equation}
		where $T_{n,z}(\eta)$ is given by 
		\begin{equation}\label{eq:Teta}
			T_{n,z}(\eta)=\frac{e^{2\pi\eta}}{q_n^{22}}\begin{bmatrix}\lambda^{-1}z^{-1} \det Q_n +\lambda'\lambda^{-1}(q_n^{21}-q_n^{12}) + z{\lambda'}^2\lambda^{-1}&e^{2\pi\eta}(q_n^{12}-\lambda' z)\\-e^{-2\pi\eta}(q_n^{21} + \lambda' z)&\lambda z\end{bmatrix}.
		\end{equation}
		\item If $z \in \bbC\setminus\{0\}$ and $Q_n$ is not off-diagonal,
		\begin{equation}\label{eq:tm:main2}
			\det T_{n,z}(\eta) = {e^{4\pi\eta}}\frac{q_n^{11}}{q_n^{22}}.
		\end{equation}			
		\item $T_{n,z}(\eta)$ and $T_{n,z}(\eta=0)$ are related by
		\begin{equation}\label{eq:Teta_T0}
			T_{n,z}(\eta)
			=e^{2\pi\eta}e^{\pi\eta\sigma_3}T_{n,z}(0)e^{-\pi\eta\sigma_3}
		\end{equation}
		where $\sigma_3$ denotes the third Pauli matrix.
	\end{enumerate}
\end{prop}

To show this, we first write $W_{\lambda,\eta}$ in coordinates.
\begin{lemma}\label{lem:w}
	Suppose $\lambda \in (0,1]$, $\eta\in\bbR$ and $W = S_{\lambda,\eta} Q$ is a split-step walk with coupling constant $\lambda$, non-reciprocality parameter $\eta$ and coins $\{Q_n\}_{n\in\bbZ}$.
	For each $n \in \bbZ$, we have
	\begin{align} \label{eq:lemW+}
		[W\psi]_n^+ & = \lambda e^{2\pi\eta}\left(q_{n-1}^{11}\psi_{n-1}^+ + q_{n-1}^{12}\psi_{n-1}^-\right)
		-\lambda' \left(q_{n}^{21}\psi_{n}^+ + q_{n}^{22}\psi_{n}^-\right),\\
		\label{eq:lemW-}
		[W\psi]_n^- & = \lambda e^{-2\pi\eta}\left(q_{n+1}^{21}\psi_{n+1}^+ + q_{n+1}^{22}\psi_{n+1}^-\right)
		+ \lambda' \left(q_{n}^{11}\psi_{n}^+ + q_{n}^{12}\psi_{n}^-\right).
	\end{align}
\end{lemma}
	
\begin{proof}
	Writing out the coordinates, we have
	\begin{align*}
		[W\psi]_n^+ & = [S_{\lambda,\eta}Q\psi]_n^+ \\ 
		& =\lambda e^{2\pi\eta}[Q\psi]_{n-1}^+ -\lambda' [Q\psi]_n^- \\
		& = \lambda e^{2\pi\eta}\left(q_{n-1}^{11}\psi_{n-1}^+ + q_{n-1}^{12}\psi_{n-1}^-\right)
		-\lambda' \left(q_{n}^{21}\psi_{n}^+ + q_{n}^{22}\psi_{n}^-\right),
	\end{align*}
	proving \eqref{eq:lemW+}. The proof of \eqref{eq:lemW-} is similar.
\end{proof}
\begin{proof}[Proof of Proposition \ref{prop:cocycle:tmMain}]
	\noindent (a) Beginning with Lemma~\ref{lem:w}, plug \eqref{eq:lemW+} into the eigenequation $W\psi=z\psi$ and shift indices $n\mapsto n+1$ to obtain
	\begin{equation}\label{eq:wpsi+}
		z\psi_{n+1}^+ = \lambda e^{2\pi\eta}\left(q_{n}^{11}\psi_{n}^+ + q_{n}^{12}\psi_{n}^-\right)
		-\lambda' \left(q_{n+1}^{21}\psi_{n+1}^+ + q_{n+1}^{22}\psi_{n+1}^-\right).
	\end{equation}
	Similarly, we obtain from \eqref{eq:lemW-}, $W\psi = z\psi$, and $n\mapsto n-1$
	\begin{equation}\label{eq:wpsi-}
		z\psi_{n-1}^- = \lambda e^{-2\pi\eta}\left(q_{n}^{21}\psi_{n}^+ + q_{n}^{22}\psi_{n}^-\right) + \lambda'\left(q_{n-1}^{11}\psi_{n-1}^+ + q_{n-1}^{12}\psi_{n-1}^-\right).
	\end{equation}
	Then, $\eqref{eq:lemW-}\cdot\lambda'e^{2\pi\eta}+\eqref{eq:wpsi+}\cdot\lambda$ and $\eqref{eq:lemW+}\cdot\lambda'e^{-2\pi\eta} - \eqref{eq:wpsi-}\cdot \lambda$ give 
	\begin{equation} \label{eq:tm:eigeqDerived1}
		\lambda'e^{2\pi\eta} z\psi_n^- + \lambda z\psi_{n+1}^+ 
		=e^{2\pi\eta}\left(\lambda^2 +{\lambda'}^2\right)\left( q_{n}^{11}\psi_{n}^+ + q_{n}^{12}\psi_{n}^-\right)
		= e^{2\pi\eta}\left(q_{n}^{11}\psi_{n}^+ + q_{n}^{12}\psi_{n}^-\right)
	\end{equation}
	and
	\begin{equation} \label{eq:tm:eigeqDerived2}
		\lambda'e^{-2\pi\eta}z \psi_n^+ - \lambda z \psi_{n-1}^-
		= -e^{-2\pi\eta}\left(\lambda^2 +{\lambda'}^2\right)\left(q_{n}^{21}\psi_{n}^+ + q_{n}^{22}\psi_{n}^-\right)
		=- e^{-2\pi\eta}\left(q_{n}^{21}\psi_{n}^+ + q_{n}^{22}\psi_{n}^-\right),
	\end{equation}
	respectively. 
	Solving \eqref{eq:tm:eigeqDerived2} for $\psi_n^-$
	yields
	\begin{equation} \label{eq:tm:eigeqDerived3} \psi_n^- = \frac{1}{q_n^{22}}\left(\lambda e^{2\pi\eta}z \psi_{n-1}^- -(q_n^{21} + \lambda' z) \psi_n^+\right),
	\end{equation}
	which is the bottom row of \eqref{eq:tm:main1}. 
	Note that this step uses the assumption that $Q_n$ is not off-diagonal, that is, $|q_n^{11}| = |q_n^{22}| \neq 0$.
	Solving \eqref{eq:tm:eigeqDerived1} for $\psi_{n+1}^+$ and inserting \eqref{eq:tm:eigeqDerived3} produces
	\begin{align*}
		\psi_{n+1}^+ 
		& = \frac{e^{2\pi\eta}}{\lambda z} \left( q_n^{11} \psi_n^+ +(q_n^{12}-\lambda' z)\psi_n^- \right) \\
		& = \frac{e^{2\pi\eta}}{\lambda z} \left( q_n^{11} \psi_n^+ +(q_n^{12}-\lambda' z) \left( \frac{1}{q_n^{22}}\left(\lambda e^{2\pi\eta}z \psi_{n-1}^- -(q_n^{21} + \lambda' z) \psi_n^+\right) \right) \right) \\
		& = \frac{e^{2\pi\eta}}{\lambda z q_n^{22}} \left(( q_n^{11}q_n^{22} - (q_n^{12}-\lambda' z)(q_n^{21} + \lambda' z) ) \psi_n^+ +(q_n^{12}-\lambda' z)\lambda e^{2\pi\eta}z \psi_{n-1}^- \right) \\
		& = \frac{1}{q_n^{22}} \left((\lambda^{-1}z^{-1} e^{2\pi\eta}\det Q_n +\lambda'\lambda^{-1}e^{2\pi\eta}(q_n^{21}-q_n^{12}) + z{\lambda'}^2\lambda^{-1}e^{2\pi\eta} ) \psi_n^+ +e^{4\pi\eta}(q_n^{12}-\lambda' z)  \psi_{n-1}^- \right),
	\end{align*}
	concluding the proof of \eqref{eq:Teta}. 
	Assertions~(b) and (c) follow directly from \eqref{eq:Teta}.
\end{proof}
We note that in the reciprocal setting $\eta=0$, the transfer matrices are unimodular, that is, $|\det T_{n,z}(0)|=1$. In the non-reciprocal setting, we have to be slightly more attentive. 
\medskip

For split-step walks with coupling constant $\lambda$ and non-reciprocality parameter $\eta$, the asymptotic behavior of generalized eigenstates of $W=S_{\lambda,\eta}Q$ in the positive lattice direction is determined by the Lyapunov exponent $L_{\lambda,\eta}^{\rm right}(z)$ defined for $z\in\bbC$ by
\begin{equation}\label{eq:Lyap_def:RIGHT_gen}
	L_{\lambda,\eta}^{\rm right}(z)=\lim_{n\to\infty}\frac1{n}\log\left\|\prod_{k=n-1}^0 T_{k,z}(\eta) \right\|,
\end{equation}
assuming that the limit exists.
We also have for iterating in the negative direction
\begin{equation}\label{eq:Lyap_def:LEFT_gen}
	L_{\lambda,\eta}^{\rm left}(z)= \lim_{n\to\infty}\frac1{n}\log\left\|\prod_{k=-n}^{-1} [T_{k,z}(\eta)]^{-1} \right\|
\end{equation}
with $T_{k,z}(\eta)$ given in \eqref{eq:Teta}.
The overall inverse localization length for the non-reciprocal model is then defined by
\begin{equation}
	L_{\lambda,\eta}(z) := \max\left\{ 0 , \min\left\{ L_{\lambda,\eta}^{\rm left}(z),L_{\lambda,\eta}^{\rm right}(z) \right\}\right\}
\end{equation}

In the reciprocal setting $\eta=0$, one has $|\det T_{n,z}| = 1$ for all $n$ and $z$, so $\|T_{n,z}^{-1}\| = \|T_{n,z}\|$ by the singular value decomposition and hence $L_{\lambda,\eta=0}^{\mathrm{left}}=L_{\lambda,\eta=0}^{\mathrm{right}}\equiv L_{\lambda,\eta=0}$.
It is straightforward to see that the non-reciprocal hopping $\eta$ contributes linearly to the Lyapunov exponent:
\begin{prop}\label{prop:Lyap_LR_eta}
	Consider the split-step walk with non-reciprocal hopping $W=S_{\lambda,\eta}Q$.
	Then, for any $\eta$ and $z$, the right-half-line Lyapunov exponent is given by
	\begin{equation}\label{eq:lyap_eta:RIGHT}
		L_{\lambda,\eta}^{\rm right}(z)=L_{\lambda,\eta=0}(z)+2\pi\eta.
	\end{equation}
	Similarly, the left-half-line Lyapunov exponent is given by
	\begin{equation}\label{eq:lyap_eta:LEFT}
		L_{\lambda,\eta}^{\rm left}(z)=L_{\lambda,\eta=0}(z)-2\pi\eta.
	\end{equation}
	so
	\begin{equation}\label{eq:lyap_eta}
		L_{\lambda,\eta}(z) = \max\left\{0 , L_{\lambda,\eta=0}(z)- 2\pi |\eta|\right\}.
	\end{equation}
\end{prop}
\begin{proof}
	This follows from the definition of the Lyapunov exponent in \eqref{eq:Lyap_def:RIGHT_gen} together with \eqref{eq:Teta_T0}:
	\begin{align*}
		L_{\lambda,\eta}^{\rm right}(z)
		& =\lim_{n\to\infty}\frac1{n}\log\left\|\prod_{k=n-1}^0 T_{k,z}(\eta) \right\| \\
		&=\lim_{n\to\infty}\frac1{n}\log\left\|\prod_{k=n-1}^0 e^{2\pi\eta}e^{\pi\eta\sigma_3}T_{k,z}(0)e^{-\pi\eta\sigma_3} \right\| \\
		&=\lim_{n\to\infty}\frac1{n}\log\left[e^{n2\pi\eta}\left\| e^{\pi\eta\sigma_3}\left[\prod_{k=n-1}^0T_{k,z}(0)\right]e^{-\pi\eta\sigma_3} \right\|\right]\\
		&=2\pi\eta+\lim_{n\to\infty}\frac1{n}\log\left\| e^{\pi\eta\sigma_3}\left[\prod_{k=n-1}^0 T_{k,z}(0)\right]e^{-\pi\eta\sigma_3} \right\|.
	\end{align*}
	Writing $X = \exp ( \pi \eta \sigma_3)$, one clearly has $\|X\|=e^{\pi|\eta|}=\|X^{-1}\|$. Combining this with the elementary inequalities
	\begin{equation*}
		\frac{\|A\|}{\|X\|\|X^{-1}\|} \leq	\|XAX^{-1}\|\leq\|X\|\|A\|\|X^{-1}\|,
	\end{equation*}
	produces
	\begin{equation*}
		L_{\lambda,\eta}^{\rm right}(z) = 2\pi\eta+\lim_{n\to\infty}\frac1{n}\log\left\| \prod_{k=1}^nT_{k,z}(0)\right\| = 2\pi \eta + L_{\lambda,0}(z),
	\end{equation*}
	as desired. The argument for $L_{\lambda,\eta}^{\rm left}$ is similar.
\end{proof}

\begin{remark}
	We will see below that localization holds as long as $\min\{L_{\lambda,\eta}^{\mathrm{right}},L_{\lambda,\eta}^{\mathrm{left}}\}>0$. The decay of the eigenstates is described by $L_{\lambda,\eta}^{\mathrm{right}}$ {\rm(}resp.\ $L_{\lambda,\eta}^{\mathrm{left}}${\rm)} to the right {\rm(}resp.\ to the left{\rm)}.
\end{remark}

\subsection{Transfer matrices and Lyapunov exponents for the PUAMO}
Let us return from the general setting to a discussion of the PUAMO with local coins
\begin{equation}\label{eq:coindefApp}
		Q_n :=Q_{n,\lambda_2,\eps}=\begin{bmatrix}\lambda_2\cos(2\pi (n\Phi+\theta+i\eps))+i\lambda_2'&-\lambda_2\sin(2\pi (n\Phi+\theta+i\eps))\\\lambda_2\sin(2\pi (n\Phi+\theta+i\eps))&\lambda_2\cos(2\pi 	(n\Phi+\theta+i\eps))-i\lambda_2'\end{bmatrix}.
\end{equation}
In this setting, abbreviating $\vartheta=\theta+i\eps$ as in the body of the paper, the transfer matrices from \eqref{eq:Teta} are given by $T_{n,z}(\eta,\vartheta)=A_z(\eta,n\Phi+\vartheta)$ where
\begin{equation}\label{eq:model:transmatDeftwopi}
	A_{z}(\eta,\vartheta)=\frac{e^{2\pi\eta}}{\lambda_2\cos(2\pi\vartheta)-i\lambda_2'}\begin{bmatrix}\lambda_1^{-1}z^{-1}+2\lambda_1'\lambda_1^{-1}\lambda_2\sin(2\pi\vartheta)+z{\lambda_1'}^2\lambda_1^{-1}	&	-e^{2\pi\eta}(\lambda_2\sin(2\pi\vartheta)+\lambda_1'z)	\\	-e^{-2\pi\eta}(\lambda_2\sin(2\pi\vartheta)+\lambda_1'z)	&	\lambda_1 z	\end{bmatrix}.
\end{equation}
In this setting (and more generally in any setting where the coins are stochastic), for $\eta=0$ the Lyapunov exponents from \eqref{eq:Lyap_def:RIGHT_gen} and \eqref{eq:Lyap_def:LEFT_gen} exist for almost every $\theta$ and are constant (in $\theta$), see for example \cite[Corollary 10.5.25]{Simon2005OPUC2}. We define
\begin{equation}\label{eq:Lyap_def:RIGHT}
	L_{\lambda,\eta,\eps}^{\rm right}(z)
	:= L(\Phi,A_z(\eta,\cdot  +i\eps))
	=\lim_{n\to\infty}\frac1{n}\int_0^1 \log\left\|\prod_{k=n-1}^0 T_{k,z}(\eta,\theta + i \eps) \right\| d\theta.
\end{equation}
We also have for iterating in the negative direction
\begin{equation}\label{eq:Lyap_def:LEFT}
	L_{\lambda,\eta,\eps}^{\rm left}(z)
	:= L(-\Phi,(A_z(\eta,\cdot +i\eps))^{-1})
	= \lim_{n\to\infty}\frac1{n}\int_0^1 \log\Bigg\|\prod_{k=-n}^{-1} [T_{k,z}(\eta,\theta + i \eps)]^{-1} \Bigg\| d\theta.
\end{equation}
As in the setting with general coins, the overall inverse localization length for the non-reciprocal PUAMO is then defined by
\begin{equation}
	L_{\lambda,\eta,\eps}(z) = \max\left\{ 0 , \min\left\{ L_{\lambda,\eta,\eps}^{\rm left}(z),L_{\lambda,\eta,\eps}^{\rm right}(z) \right\}\right\}.
\end{equation}

Since we already discussed the contribution of the non-reciprocal hopping parameter in Proposition \ref{prop:Lyap_LR_eta}, let us set $\eta=0$ for the following discussion.
It is straightforward to see from \eqref{eq:tm:main2} that one might push the transfer matrices $T_{n,z}$ of the PUAMO into the unimodular setting by considering $T_{n,z}/\sqrt{\det(T_{n,z})}$ as long as $T_{n,z}$ is invertible. However, the attentive reader might have noticed that the general theory from Appendix \ref{app:global_theory} can still not directly be applied to $A_z(\eta,\vartheta)$ defined in \eqref{eq:model:transmatDeftwopi} since in Appendix \ref{app:global_theory} also analyticity is required. As discussed in \cite[Section 4]{CFO1}, analyticity might be restored by studying instead the regularized transfer matrix cocycle
\begin{equation*}
	B_{z}(0,\vartheta)=\left[\frac{2(\lambda_2\cos(2\pi\vartheta)-i\lambda_2')}{1+\lambda_2'}\right]A_{z}(0,\vartheta)
\end{equation*}
The relation between the Lyapunov exponents of $A_z$ and $B_z$ can be inferred from studying the prefactor (see \cite[Lemma 4.7]{CFO1}) and is given for all $\eps\in\bbR$ by \cite[(4.34)]{CFO1}
\begin{equation*}
	L(B_z,\eps)=L(A_z,\eps)+2\pi\max\{|\eps|-\eps_0,0\},
\end{equation*}
such that in particular, $L(B_z,\eps)=L(A_z,\eps)$ whenever $|\eps|\leq\eps_0$. Here, $\eps_0$ is as defined in the body of the paper, that is,
\begin{equation*}
	\eps_0=\frac1{2\pi}\sinh^{-1}\left[\frac{\lambda_2'}{\lambda_2}\right]=\frac1{2\pi}\log\left[\frac{1+\lambda_2'}{\lambda_2}\right].
\end{equation*}

Since $B_z$ is analytic and admits an analytic extension to $\bbC$, one can compute its Lyapunov exponent using global theory \cite{CFO1}, see also Appendix \ref{app:global_theory}. The strategy here is to calculate $L(B_z,\eps)$ in the limit $\eps\to\infty$ and then infer the global behaviour from continuity, convexity and quantization of acceleration. Then, employing the relation between $L(B_z,\eps)$ and $L(A_z,\eps)$ given above, one obtains the following explicit formula for the Lyapunov exponent of $A_z$ for $z$ in the spectrum of $W_{\lambda_1,\lambda_2}$: writing $L(A_z,\eps)=L_{\lambda_1,\lambda_2,\eta=0,\eps}(z)$ we have that
\begin{equation}\label{eq:lyap_res}
	L_{\lambda_1,\lambda_2,\eta=0,\eps}(z)=\max\{2\pi|\eps|+\log\lambda_0-2\pi\max\{|\eps|-\eps_0,0\},0\},
\end{equation}
where
\begin{equation*}
	\lambda_0=\frac{\lambda_2(1+\lambda_1')}{\lambda_1(1+\lambda_2')}
\end{equation*}
plays a role analogous to that of $\lambda$ in the almost-Mathieu operator.

The value $\eps=\eps_0$ up to which the Lyapunov exponents of $A_z$ and $B_z$ agree corresponds exactly to the value at which $A_z$ is not analytic. Combining \eqref{eq:lyap_res} with global theory as in Appendix \ref{app:global_theory} one can establish that all spectrum is off the unit circle after this second phase transition point: 
\begin{theorem} \label{t:2ndtransition}
	Let $\lambda_1>\lambda_2$ and $z\in\partial \bbD$. Then, for all $|\eps|>\eps_0$ and $\eta=0$, $z$ is not in the spectrum of $W_{\lambda_1,\lambda_2,\eta=0,\eps}$.
\end{theorem}
\begin{proof}
	Fixing $z \in \partial \bbD$ and $\eps>\eps_0$, \eqref{eq:lyap_res} implies that in the reciprocal case $\eta=0$ we have $L_{\lambda_1,\lambda_2,\eta=0,\eps}(z)>0$ 
	and $\partial_\eps L_{\lambda_1,\lambda_2,\eta=0,\eps}(z)=0$. Using  the discussion in Appendix~\ref{app:global_theory}, one sees that $z$ is not in the spectrum by adapting the standard theory of CMV matrices as in \cite{FOZ2017CMP} to the case of complexified phases; the details are beyond the scope of this manuscript and will be carried out in a future mathematical work of the authors.
\end{proof}

Following computations similar to that in the proof of Proposition \ref{prop:cocycle:tmMain}, one finds that the transfer matrices of the transposed split-step walk $W^\top = Q^\top S_{\lambda,\eta}^\top$ are given by
\begin{equation}\label{eq:TetaT}
	\top_{n,z}(\eta)=\frac{e^{-2\pi\eta}}{q_n^{11}}\begin{bmatrix}\lambda^{-1}z  +\lambda'\lambda^{-1}(q_n^{21}-q_n^{12}) + z^{-1}{\lambda'}^2\lambda^{-1}\det Q_n&-e^{-2\pi\eta}(q_n^{21}+\lambda' z^{-1}\det Q_n)\\e^{2\pi\eta}(q_n^{12} - \lambda' z^{-1}\det Q_n)&\lambda z^{-1}\det Q_n\end{bmatrix}.
\end{equation}
Applying this to the model at hand by plugging in the coin parameters from \eqref{eq:coindefApp}, the transfer matrices for the dual PUAMO $W_{\lambda_1,\lambda_2,\eta,\eps}^\aubrydual=Q_{\lambda_1,-\eta}^{\top}S_{\lambda_2,-\eps}^\top$ are given by $T_{\lambda_1,\lambda_2,n,z}(\eta,\theta+i\eps)^\aubrydual=\overline{T_{\lambda_2,\lambda_1,n,1/\overline z}(\eps,\theta+i\eta)}$. 
Hence, we conclude from Proposition \ref{prop:Lyap_LR_eta} that $\eps$, which plays the role of a non-reciprocal hopping parameter for the dual model analogous to that of  $\eta$ for the original model, contributes linearly to $L^\aubrydual_{\lambda_1,\lambda_2,\eta,\eps}(z)$. Moreover, from \cite{CFO1} resp. from \eqref{eq:lyap_res} we obtain for $\eps=0$
\begin{equation}\label{eq:lyap_dual_app}
	L^\aubrydual_{\lambda_1,\lambda_2,\eta,\eps=0}(z)=\max\{2\pi|\eta|-\log\lambda_0-2\pi\max\{|\eta|-\eta_0,0\},0\},
\end{equation}
where we set $\eta_0=\sinh^{-1}(\lambda_1'/\lambda_1)/(2\pi)$.

Thus, for $z$ belonging to the spectrum (of the dual model), the right and left dual Lyapunov exponents are given by
\begin{equation}\label{eq:lyap_dual_app}
	L^{\aubrydual,\textrm{right}}_{\lambda_1,\lambda_2,\eta,\eps}(z)=2\pi\eps+\max\{2\pi|\eta|-\log\lambda_0-2\pi\max\{|\eta|-\eta_0,0\},0\},
\end{equation}
and $L^{\aubrydual,\textrm{left}}_{\lambda_1,\lambda_2,\eta,\eps}(z)=L^{\aubrydual,\textrm{right}}_{\lambda_1,\lambda_2,\eta,-\eps}(z)$, respectively. 

As discussed above, the PUAMO $W_{\lambda_1,\lambda_2,\eta,\eps}$ localizes as soon as $L_{\lambda_1,\lambda_2,\eta,\eps}(z) = L_{\lambda_1,\lambda_2,\eta=0,\eps}(z)- 2\pi |\eta|>0$. Accordingly, the critical point of the first phase transition is given by
\begin{equation*}
	2\pi|\eta|=\max\{2\pi|\eps|+\log\lambda_0-2\pi\max\{|\eps|-\eps_0,0\},0\},
\end{equation*}
which provides the phase boundaries in Figure \ref{fig:eta_eps_phasediagram}.

\section{Skin transformations}\label{app:skin}

We first note that a split-step walk with non-reciprocal hopping is always similar to a walk with reciprocal hopping, albeit via an unbounded similarity transformation.
\begin{lemma}\label{lem:skin_trafo}
	Let $W_{\lambda,\eta}=S_{\lambda,\eta}Q$ be a split-step walk with non-reciprocal hopping parameter $\eta$ and let $Q$ be a general coin. Then
	\begin{equation*}
		W_{\lambda,\eta}=H_\eta^{-1}W_{\lambda}H_\eta
	\end{equation*}
	where $W_{\lambda}=W_{\lambda,0}$ and the similarity transform is defined by $[H_\eta\psi]_n^\pm=e^{2\pi\eta n}\psi_n^\pm$.
\end{lemma}
\begin{proof}
	Consider the coordinate expressions of the action of $W_{\lambda,\eta}$ on $\psi_n^+$ given in \eqref{eq:lemW+}. We directly infer that
	\begin{align*}
		[H_\eta^{-1}W_{\lambda,\eta}H_\eta\psi]_n^+&=e^{-2\pi\eta n}[WH_\eta\psi]_n^+ \\
		& = e^{-2\pi\eta n}\Big[\lambda e^{2\pi\eta}\left(q_{n-1}^{11}[H_\eta\psi]_{n-1}^+ + q_{n-1}^{12}[H_\eta\psi]_{n-1}^-\right)
		-\lambda' \left(q_{n}^{21}[H_\eta\psi]_{n}^+ + q_{n}^{22}[H_\eta\psi]_{n}^-\right)\Big]\\
		& = e^{-2\pi\eta n}\Big[\lambda e^{2\pi\eta}\left(q_{n-1}^{11}e^{2\pi\eta (n-1)}\psi_{n-1}^+ + q_{n-1}^{12}e^{2\pi\eta (n-1)}\psi_{n-1}^-\right)\\
		&\qquad-\lambda' \left(q_{n}^{21}e^{2\pi\eta n}\psi_{n}^+ + q_{n}^{22}e^{2\pi\eta n}\psi_{n}^-\right)\Big]\\
		& = \lambda \left(q_{n-1}^{11}\psi_{n-1}^+ + q_{n-1}^{12}\psi_{n-1}^-\right)
		-\lambda' \left(q_{n}^{21}\psi_{n}^+ + q_{n}^{22}\psi_{n}^-\right)\\
		&=[W_{\lambda,0}\psi]_n^+.
	\end{align*}
	The statement for $[W\psi]_n^-$ follows analogously from \eqref{eq:lemW-}.
\end{proof}
This lets one infer that if $\psi$ solves the (generalized) eigenequation $W_{\lambda}\psi=z\psi$, then $H_\eta^{-1}\psi=(e^{-2\pi \eta n}\psi_n)_{n\in\bbZ}$ solves the (generalized) eigenequation for $W_{\lambda,\eta}$. In particular, if $\psi$ is (exponentially) localized around some localization center $n_0$ with inverse localization length/Lyapunov exponent $\ell>0$, i.e., $\psi_n\sim e^{-\ell|n-n_0|}$, then 
\begin{equation*}
	(H_\eta^{-1}\psi)_n\sim \begin{cases}e^{-(\ell+2\pi\eta)|n-n_0|}& n>n_0\\e^{-(\ell-2\pi\eta)|n-n_0|}& n<n_0.\end{cases}
\end{equation*}
Thus, as long as $2\pi|\eta|< \ell$ everything remains localized. However, as soon as $2\pi|\eta|> \ell$, delocalization occurs either on the left ($\eta>0$) or on the right ($\eta<0$) of $n_0$. 
This is captured by the right and the left Lyapunov exponents in \eqref{eq:Lyap_def:RIGHT} and \eqref{eq:Lyap_def:LEFT}, respectively: for $2\pi\eta>L_{\lambda,\eta=0,\eps}(z)>0$ the right Lyapunov exponent is positive while the left is negative and vice versa when $-2\pi\eta>L_{\lambda,\eta=0,\eps}(z)>0$; compare with Proposition \ref{prop:Lyap_LR_eta}. 
\medskip

Consider a walk on a finite lattice of size $N$ indexed by $n\in\{0,\dots,N-1\}$ which is defined as a product of a shift and a coin with periodic boundary conditions, i.e.,
\begin{equation*}
	W_{\lambda,\eta,N}=S_{\lambda,\eta,N} Q_N,
\end{equation*}
where, with the obvious modifications for open boundary conditions,
\begin{equation}\label{eq:shift_def_finite}
	\big[S_{\lambda,\eta,N}\psi\big]_{n}^\pm=\lambda e^{2\pi \eta}\lambda\psi_{(n\mp1)\bmod N}^\pm\mp\lambda'\psi_n^\mp,\qquad Q_N=\bigoplus_{n=0}^{N-1}Q_n.
\end{equation}
Then the lemma above implies that
\begin{lemma}\label{lem:W_corner_els}
	Let $S_{\lambda_1,\eta,N}$ be the finite-size shift defined above and $Q_N=\bigoplus_{n=0}^{N-1}Q_n$ be a finite-dimensional coin. Then  $W_{\lambda,\eta,N}=S_{\lambda,\eta,N}Q_N$ is similar to $\widetilde W_{\lambda,\eta,N}=\widetilde S_{\lambda,\eta,N}Q_N$ where
	\begin{equation*}
		\Big[\widetilde S_{\lambda,\eta,N}\psi\Big]_{n}^\pm=e^{2\pi ((n\mp1)\bmod N-(n\mp1))\eta}\lambda\psi_{(n\mp1)\bmod N}^\pm\mp\lambda'\psi_n^\mp.
	\end{equation*}
\end{lemma}
Note that in the finite-dimensional setting the skin-transformed walk still depends on $\eta$ at the corners. In contrast, in the infinite-dimensional case of Lemma \ref{lem:skin_trafo} the skin-transformed walk is independent of $\eta$.
\begin{proof}
	Let $H_{\eta,N}$ be the finite-dimensional version of the similarity transform from Lemma \ref{lem:skin_trafo}. Since $H_{\eta,N}$ is block-diagonal, it commutes with the coin $Q_N$. Moreover, from \eqref{eq:shift_def_finite} we have that
	\begin{align*}
		[H_{\eta,N}^{-1}S_{\lambda,\eta,N}H_{\eta,N}\psi]_n^+&=e^{-2\pi n\eta}[S_{\lambda,\eta,N}H_{\eta,N}\psi]_n^+\\
		&=e^{-2\pi n\eta}\Big[e^{2\pi\eta}\lambda[H_{\eta,N}\psi]_{(n-1)\bmod N}^+-\lambda'[H_{\eta,N}\psi]_n^-\Big]\\
		&=e^{-2\pi n\eta}\Big[e^{2\pi\eta}e^{2\pi ((n-1)\bmod N)\eta}\lambda\psi_{(n-1)\bmod N}^+-\lambda'e^{2\pi n\eta}\psi_n^-\Big]\\
		&=e^{2\pi ((n-1)\bmod N-(n-1))\eta}\lambda\psi_{(n-1)\bmod N}^+-\lambda'\psi_n^-.
	\end{align*}
	Similarly, we obtain
	\begin{equation*}
		[H_{\eta,N}^{-1}S_{\lambda,\eta,N}H_{\eta,N}\psi]_n^- =e^{2\pi ((n+1)\bmod N-(n+1))\eta}\lambda\psi_{(n+1)\bmod N}^-+\lambda'\psi_{n}^+.
	\end{equation*}
\end{proof}
We will apply this transformation below to compute the characteristic polynomial of finite-dimensional approximations of the UAMO. Concretely, we there utilize that by Lemma \ref{lem:W_corner_els} everything but the corner elements is independent of $\eta$, and that the elements in one of the corners are suppressed exponentially in $N$.

\section{Finite-dimensional Aubry duality}

Next, let us consider finite-size approximations to $W_{\lambda_1,\lambda_2}$ obtained from the rational approximation $\Phi_k=\nr_k/\dr_k$ to $\Phi$ in terms of continued fractions with $\dr_k=N$.
In this setting, we have the following finite-dimensional version of Aubry duality:
\begin{lemma}\label{lem:Aubry_dual_fin}
	Let $\Phi$ be rational with denominator $N \in\bbN$. Moreover, let		
	$S_{\lambda_1,\eta,N}$ be the finite-size shift defined in \eqref{eq:shift_def_finite} and $Q_{\lambda_2,\Phi,\theta,\eps,N}$ be the finite-size coin which locally acts as in \eqref{eq:coindefApp}. Then, for every $\xi\in\frac1N\bbZ$, $W_{\lambda_1,\lambda_2,\eta,\eps,\Phi,\theta,N}=S_{\lambda_1,\eta,N}Q_{\lambda_2,\Phi,\theta,\eps,N}$ is unitarily equivalent {\rm(}``Aubry-dual''{\rm)} to $W_{\lambda_2,\lambda_1,i\theta-\eps,-\eta,\Phi,\xi, N}^\top=Q_{\lambda_1,\Phi,\xi,-\eta,N}^\top S_{\lambda_2,i\theta-\eps,N}^\top$	where the transposed shift is given by
	\begin{equation*}
		\left[S_{\lambda_2,\vartheta,N}^\top\psi\right]_{n}^\pm=\lambda_2e^{\pm2\pi i \vartheta}\psi_{(n\pm1)\bmod N}^\pm\pm\lambda_2'\psi_n^\mp.
	\end{equation*}
	In particular, if $\psi\in\bbC^N\otimes\bbC^2$ solves the eigenequation $W_{\lambda_1,\lambda_2,\eta,\eps,\Phi,\theta,N}\psi=z\psi$, then $\varphi=(\varphi_n)_{n\in\bbZ}$ with components
	\begin{equation*}
		\begin{bmatrix}\varphi_n^+\\\varphi_n^-\end{bmatrix}=\frac{e^{2\pi in\theta}}{\sqrt 2}\begin{bmatrix}1&i\\i&1\end{bmatrix}\begin{bmatrix}\check \psi^+(\Phi n+\xi)\\\check \psi^-(\Phi n+\xi)\end{bmatrix}
	\end{equation*}
	solves the eigenequation for $W_{\lambda_2,\lambda_1,-\eps,-\eta,\Phi,\xi,N}^\top$, where $\check\psi(n)^\pm=\frac{1}{\sqrt{N}}\sum_{m=0}^{N-1}e^{2\pi imn}\psi_m^\pm$ denotes the inverse {\rm(}discrete{\rm)} Fourier transform of $\psi$.
\end{lemma}
\begin{proof}
	Let us consider the following transformation
	\begin{align*}
		[U(\xi)\psi]_n^+=\frac1{\sqrt{2N}}\sum_{m=0}^{N-1}e^{2\pi i m(n\Phi+\xi)}(\psi_m^++i\psi_m^-),\qquad
		[U(\xi)\psi]_n^-=\frac1{\sqrt{2N}}\sum_{m=0}^{N-1}e^{2\pi i m(n\Phi+\xi)}(i\psi_m^++\psi_m^-).
	\end{align*}
	It is straightforward to verify that $U(\xi)$ is unitary and that the action of its inverse is given by
	\begin{align*}
		[U(\xi)^{-1}\psi]_n^+=\frac1{\sqrt{2N}}\sum_{m=0}^{N-1}e^{-2\pi i n(m\Phi+\xi)}(\psi_m^+-i\psi_m^-),\qquad
		[U(\xi)^{-1}\psi]_n^-=\frac1{\sqrt{2N}}\sum_{m=0}^{N-1}e^{-2\pi i n(m\Phi+\xi)}(-i\psi_m^++\psi_m^-).
	\end{align*}
	Combining this with the action of the finite-size shift from \eqref{eq:shift_def_finite} one verifies by direct calculation that for $\xi \in \frac1N \bbZ$:
	\begin{align*}
		[U(\xi)S_{\lambda_1,\eta,N}U(\xi)^{-1}\psi]_n^+
				&=\frac1{2{N}}\Big(\lambda_1e^{2\pi\eta}\sum_{m=0}^{N-1}\sum_{\ell=0}^{N-1}e^{2\pi i m(n\Phi+\xi)}e^{-2\pi i((m-1)\bmod N)(\ell\Phi+\xi)}[\psi_\ell^+-i\psi_\ell^-]\\
				&\qquad-\lambda_1'\sum_{m=0}^{N-1}\sum_{\ell=0}^{N-1}e^{2\pi i m(n-\ell)\Phi}[-i\psi_\ell^++\psi_\ell^--i\psi_\ell^+-\psi_\ell^-]\\
				&\qquad+i\lambda_1e^{-2\pi\eta}\sum_{m=0}^{N-1}\sum_{\ell=0}^{N-1}e^{2\pi i m(n\Phi+\xi)}e^{-2\pi i((m+1)\bmod N)(\ell\Phi+\xi)}[-i\psi_\ell^++\psi_\ell^-]\Big)\\
				&=\frac1{2{N}}\Big(\lambda_1e^{2\pi\eta}\sum_{m=0}^{N-1}\sum_{\ell=0}^{N-1}e^{2\pi i m(n\Phi+\xi)}e^{-2\pi i(m-1)(\ell\Phi+\xi)}[\psi_\ell^+-i\psi_\ell^-]\\
				&\qquad+\lambda_1'2i N\psi_n^+\\
				&\qquad+i\lambda_1e^{-2\pi\eta}\sum_{m=0}^{N-1}\sum_{\ell=0}^{N-1}e^{2\pi i m(n\Phi+\xi)}e^{-2\pi i(m+1)(\ell\Phi+\xi)}[-i\psi_\ell^++\psi_\ell^-]\Big)\\
				&=\frac1{2{N}}\Big(\sum_{m=0}^{N-1}\sum_{\ell=0}^{N-1}e^{2\pi i m(n\Phi+\xi)}e^{-2\pi im(\ell\Phi+\xi)}\Big[e^{2\pi i(\ell\Phi+\xi)}\lambda_1e^{2\pi\eta}[\psi_\ell^+-i\psi_\ell^-]\\
				&\qquad+e^{-2\pi i(\ell\Phi+\xi)}\lambda_1e^{-2\pi\eta}[\psi_\ell^++i\psi_\ell^-]\Big]+\lambda_1'2i N \psi_n^+\Big)\\
				&=\Big[\lambda_1\cos(2\pi(n\Phi+\xi-i\eta))+i\lambda_1'\Big]\psi_n^++\lambda_1\sin(2\pi(n\Phi+\xi-i\eta))\psi_n^-,
	\end{align*}
	where one has to be careful with the corner elements due to the modular arithmetic in the definition of the finite-size shift.
	Similarly, we obtain
	\begin{equation*}
		[U(\xi)S_{\lambda_1,\eta,N}U(\xi)^{-1}\psi]_n^-=-\lambda_1\sin(2\pi(n\Phi+\xi-i\eta))\psi_n^++\Big[\lambda_1\cos(2\pi(n\Phi+\xi-i\eta))-i\lambda_1'\Big]\psi_n^-.
	\end{equation*}
	Thus, for every $\xi\in\frac1N\bbZ$ we have that $U(\xi)S_{\lambda_1,\eta,N}U(\xi)^{-1}=Q_{\lambda_1,\Phi,\xi-i\eta}^\top$. 
	\smallskip
	
	Applying the same steps to the finite-size coin we get
	\begin{align*}
		[U(\xi)Q_{\lambda_2,\Phi,\vartheta,N}U(\xi)^{-1}\psi]_n^+&=\frac1{2{N}}\sum_{m=0}^{N-1}e^{2\pi i m(n\Phi+\xi)}\sum_{\ell=0}^{N-1}e^{-2\pi i m(\ell\Phi+\xi)} \bigg[ \Big((\lambda_2\cos(2\pi(\Phi m+\vartheta))+i\lambda_2')[\psi_\ell^+-i\psi_\ell^-]\\
		&\qquad-\lambda_2\sin(2\pi(n\Phi+\vartheta))[-i\psi_\ell^++\psi_\ell^-]\Big)\\
		&\qquad+i\Big(\lambda_2\sin(2\pi(n\Phi+\vartheta))[\psi_\ell^+-i\psi_\ell^-]\\
		&\qquad+(\lambda_2\cos(2\pi(\Phi m+\vartheta))-i\lambda_2')[-i\psi_\ell^++\psi_\ell^-]\Big) \bigg]\\
		&=\frac1{2{N}}\sum_{m=0}^{N-1}e^{2\pi i m(n\Phi+\xi)}\sum_{\ell=0}^{N-1}e^{-2\pi i m(\ell\Phi+\xi)}\Big((2\lambda_2\cos(2\pi(\Phi m+\vartheta))\psi_\ell^++2\lambda_2'\psi_\ell^-\\
		&\qquad+2\lambda_2\sin(2\pi(n\Phi+\vartheta))i\psi_\ell^+\Big)\\
		&=\frac1{2{N}}\sum_{m=0}^{N-1}e^{2\pi i m(n\Phi+\xi)}\sum_{\ell=0}^{N-1}e^{-2\pi i m(\ell\Phi+\xi)}\Big(2\lambda_2e^{2\pi i(\Phi m+\vartheta)}\psi_\ell^++2\lambda_2'\psi_\ell^-\Big)\\
		&=\frac1{{N}}\Big(\lambda_2\sum_{m=0}^{N-1}\sum_{\ell=0}^{N-1}e^{2\pi i (mn\Phi-m(\ell-1)\Phi+\vartheta)}\psi_\ell^++\lambda_2'N\psi_n^-\Big)\\
		&=\frac1{{N}}\Big(\lambda_2\sum_{m=0}^{N-1}\sum_{\ell=0}^{N-1}e^{2\pi i (mn\Phi-m\ell\Phi+\vartheta)}\psi_{(\ell+1)\bmod N}^++\lambda_2'N\psi_n^-\Big)\\
		&=\lambda_2e^{2\pi i \vartheta}\psi_{(n+1)\bmod N}^++\lambda_2'\psi_n^-.
	\end{align*}
	Similarly,
	\begin{align*}
		[U(\xi)Q_{\lambda_2,\Phi,\vartheta,N}U(\xi)^{-1}\psi]_n^-	&=\lambda_2e^{-2\pi i\vartheta}\psi_{(n-1)\bmod N}^--\lambda_2'\psi_n^+.
	\end{align*}
	Thus, we have that $U(\xi)Q_{\lambda_2,\Phi,\vartheta,N}U(\xi)^{-1}=S_{\lambda_2,\Phi,\vartheta,N}^\top$ for every $\xi\in\frac1N\bbZ$. The second statement follows from a direct computation.
\end{proof}
	
\section{The winding number formula}
\label{app:winding}

In this section, we explain how the winding number defined by
\begin{equation*}
	\wind_\eps(z)=\lim_{N\to\infty}\frac1N\frac1{2\pi i}\int_0^1 d\theta\:\partial_\theta\log\det(W_N(\theta+i\eps)-z\idty)
\end{equation*}
detects the first phase transition in the reciprocal setting. Concretely, we show that
\begin{prop}\label{prop:wind}
	Consider the PUAMO $W_{\lambda_1,\lambda_2,\eta=0,\eps}$ in the reciprocal, $\PT$-symmetric setting. Then, for $|\eps|<L^\aubrydual(z)/(2\pi)$, $\wind_\eps(z)=0$ for all $z$ in a gap of the spectrum of the UAMO $W_{\lambda_1,\lambda_2}$. Moreover, for $\eps<-L^\aubrydual(z)/(2\pi)$ and $\eps>L^\aubrydual(z)/(2\pi)$ we have $\wind_\eps(z)=1$ and $\wind_\eps(z)=-1$, respectively, for some $z$ in a gap of the spectrum of $W_{\lambda_1,\lambda_2}$.
\end{prop}

To begin the proof, pick $z\in\partial\bbD$ not in the spectrum of the UAMO $W_{\lambda_1,\lambda_2}$ (which is a Cantor set \cite{CFO1}) and observe that
\begin{align}
	\wind_\eps(z)
	=\lim_{N\to\infty}\frac1N\frac1{2\pi i}\int_0^1 \partial_\theta\log f_N(\theta + i\eps) \, d\theta
	=\lim_{N\to\infty}\frac1N\frac1{2\pi i}\int_0^1 \frac{\partial_\theta f_N(\theta+i\eps)}{f_N(\theta+i\eps)} \, d\theta,\label{eq:winding_num_f}
\end{align}
where we abbreviate $f_N(\vartheta)=\det(W_N(\vartheta)-z\idty)$. Here, $W_N = W_{\lambda_1, \lambda_2, \eta,\eps,\Phi,\theta,N}$ is the finite-size approximation to $W_{\lambda_1, \lambda_2, \eta,\eps,\Phi,\theta}$ obtained from the rational approximation $\Phi_k=\nr_k/\dr_k$ to $\Phi$ in terms of continued fractions with $\dr_k=N$. The matrix representation of $W_N$ takes the form
\begin{equation}\label{eq:WN}
	\begin{aligned}
	&W_N=\\
	&\quad\begin{bmatrix}
			-\lambda_1' q_0^{21} & -\lambda_1' q_0^{22} &&&&&\lambda_1 e^{2\pi \eta} q_{N-1}^{11} &  \lambda_1 e^{2\pi \eta} q_{N-1}^{12}\\
			\lambda_1' q_0^{11} & \lambda_1' q_0^{12} &  \lambda_1 e^{-2\pi \eta} q_1^{21} &  \lambda_1 e^{-2\pi \eta} q_1^{22} \\
			\lambda_1 e^{2\pi \eta}q_0^{11} &  \lambda_1 e^{2\pi \eta} q_0^{12} & -\lambda_1' q_1^{21} &  -\lambda_1' q_1^{22} \\
			&&  \lambda_1' q_1^{11} &  \lambda_1' q_1^{12} \\
			&&  \lambda_1 e^{2\pi \eta} q_1^{11} &  \lambda_1 e^{2\pi \eta} q_1^{12} \\
			&&& \ddots & \ddots & \ddots\\
			&&&&&&\lambda_1 e^{-2\pi \eta} q_{N-1}^{21} &  \lambda_1 e^{-2\pi \eta} q_{N-1}^{22}\\
			&&&&&&-\lambda_1' q_{N-1}^{21} & -\lambda_1' q_{N-1}^{22}\\
			\lambda_1 e^{-2\pi \eta} q_0^{21} &  \lambda_1 e^{-2\pi \eta} q_0^{22} &&&&&\lambda_1' q_{N-1}^{11} & \lambda_1' q_{N-1}^{12}
		\end{bmatrix}.
	\end{aligned}
\end{equation}
Applying to this matrix with $\eta=0$ first the finite-dimensional Aubry duality $U(\xi)$ from Lemma \ref{lem:Aubry_dual_fin} with $\xi=0$ produces $Q_{\lambda_1,\Phi,\xi=0,0,N}^\top S_{\lambda_2,i\theta-\eps,N}^\top$ with non-reciprocal hopping parameter $-\eps$. Then, applying the finite-dimensional skin transformation from Lemma \ref{lem:W_corner_els}, we find that $W_{N}$ is similar to $\widetilde W_N=\widetilde W_{\lambda_1,\lambda_2,\eta,0,\Phi,\theta,N}$ given by
\begin{equation}\label{eq:WN_tilde}
	\begin{aligned}
		&\widetilde W_N=\\
	&\quad\begin{bmatrix}
		-\lambda_2' \widetilde{q}_0^{21} & \lambda_2' \widetilde{q}_0^{11} &\lambda_2 \widetilde{q}_0^{11}&&&& &  e^{-2\pi iN(\theta+i\eps)}\lambda_2  \widetilde{q}_0^{21}\\
		-\lambda_2' \widetilde{q}_0^{22}& \lambda_2' \widetilde{q}_0^{12} & \lambda_2  \widetilde{q}_0^{12}  &  &&&&e^{-2\pi iN(\theta+i\eps)}\lambda_2  \widetilde{q}_0^{22} \\
		& \lambda_2  \widetilde{q}_2^{21}  & -\lambda_2' \widetilde{q}_2^{21} & \lambda_2' \widetilde{q}_2^{11} & \lambda  \widetilde{q}_2^{11} \\
		&\lambda_2  \widetilde{q}_2^{22}&  -\lambda_2' \widetilde{q}_2^{22} &  \lambda_2' \widetilde{q}_2^{12} & \lambda_2  \widetilde{q}_2^{12} \\
		&&   &   \\
		&&& \ddots & \ddots & \ddots\\
		&&&&&& &  \\
		e^{2\pi iN(\theta+i\eps)}\lambda_2  \widetilde{q}_{N-1}^{11}&&&&&\lambda_2  \widetilde{q}_{N-1}^{21}&-\lambda_2' \widetilde{q}_{N-1}^{21} & \lambda_2' \widetilde{q}_{N-1}^{11}\\
		e^{2\pi iN(\theta+i\eps)}\lambda_2  \widetilde{q}_{N-1}^{12}&   &&&&\lambda_2  \widetilde{q}_{N-1}^{22}& -\lambda_2' \widetilde{q}_{N-1}^{22}& \lambda_2' \widetilde{q}_{N-1}^{12},
	\end{bmatrix}
	\end{aligned}
\end{equation}
where $\widetilde{q}_n^{11}=\lambda_1\cos(2\pi\Phi n)+i\lambda_1'=\overline{\widetilde{q}_n^{22}}$ and $\widetilde{q}_n^{12}=-\lambda_1\sin(2\pi\Phi n)=-\widetilde{q}_n^{12}$.

This similarity of $W_N$ and $\widetilde W_N$ immediately implies that 
\begin{equation*}
	f_N(\theta+i\eps)=\det(W_N(\theta+i\eps)-z\idty)=\det(\widetilde W_N(\theta+i\eps)-z\idty).
\end{equation*}
To calculate this determinant, we use a trick from the theory of CMV matrices \cite{simonCMVMatricesFive2007}. Each split-step walk $W$ is gauge-equivalent to a CMV matrix \cite{CFO1,CFLOZ}, and can hence be factorized as
\begin{equation*}
	W=(S_\lambda\Sigma_1)(\Sigma_1Q)=:\mathcal L\mathcal M,
\end{equation*}
where $\Sigma_1=\bigoplus_{n\in\bbZ}\sigma_1$ is block-diagonal with respect to the cell structure. Thus, $\mathcal M=\Sigma_1Q$ is related to the coin, whereas $\mathcal L=S_{\lambda}\Sigma_1$ is block-diagonal on ``tilted'' cells spanned by $\{\delta_{n-1}^-,\delta_n^+\}$. In the current setting, the blocks of $\mathcal L$ and $\mathcal M$ are thus given by
\begin{equation}\label{eq:LMblocks}
	\mathcal L_n=\begin{bmatrix}\lambda_1'&\lambda_1\\\lambda_1&-\lambda_1'\end{bmatrix},\qquad \mathcal M_n=\begin{bmatrix}\lambda_2\sin(2\pi(\Phi n+\theta))&\lambda_2\cos(2\pi(\Phi n+\theta))-i\lambda_2'\\\lambda_2\cos(2\pi(\Phi n+\theta))+i\lambda_2'&-\lambda_2\sin(2\pi(\Phi n+\theta))\end{bmatrix}.
\end{equation} 
Similarly, we can factorize the Aubry dual as 
\begin{equation*}
	W^\aubrydual=Q^\top S^\top=(\Sigma_1Q)^\top(S_\lambda\Sigma_1)^\top=(Q^\top\Sigma_1)(\Sigma_1S_\lambda^\top)=:\mathcal M^\top \mathcal L^\top.
\end{equation*}

This factorization greatly simplifies the calculation of the characteristic polynomial of the finite-dimensional matrix $W_N$ from \eqref{eq:WN} as well as $\widetilde W_N$: we observe that \begin{equation*}
	f_N(\theta+i\eps)=\det(W_N-z\idty)=\det\left(\mathcal L_N-z\mathcal M_N^{-1}\right)\det(\mathcal M_N)=(-1)^{N}\det\left(\mathcal L_N-z\mathcal M_N^{-1}\right),
\end{equation*} 
where $\mathcal L_N$ and $\mathcal M_N$ are specified in terms of the blocks \eqref{eq:LMblocks} with appropriate boundary conditions. The advantage is that now we reduced the problem to calculating the determinant of $\mathcal L_N-z\mathcal M_N^{-1}$, which up to the corner elements is a tridiagonal matrix, i.e.,
\begin{equation}
	\begin{aligned}
	&\mathcal L_N-z\mathcal M_N^{-1}=\\
	&\qquad\begin{bmatrix}
		-\lambda_1'+zq_0^{21}	& zq_0^{22}&&&&&\lambda_1\\
		zq_0^{11}&-zq_0^{12}+\lambda_1'&\lambda_1\\
		&\lambda_1&-\lambda_1'+zq_2^{21}&zq_2^{22}&\\
		&&zq_2^{11}&-zq_2^{12}+\lambda_1'&\lambda_1&\\
		&&&\ddots&\ddots&\ddots&\\
		&&&&&-\lambda_1'+zq_{N-1}^{21}&zq_{N-1}^{22}\\
		\lambda_1&&&&&zq_{N-1}^{11}&\lambda_1'-zq_{N-1}^{12}
	\end{bmatrix}.
	\end{aligned}
\end{equation}
Similarly, writing $\widetilde W_N=:\widetilde{\mathcal M}_N^\top\widetilde{\mathcal L}_N^\top$ for $\widetilde W_N$ as in \eqref{eq:WN_tilde} with appropriate boundary conditions gives
\begin{equation*}
	\begin{aligned}
	&\widetilde{\mathcal L}_N^\top-z(\widetilde{\mathcal M}_N^\top)^{-1}=\\
	&\qquad\begin{bmatrix}
		-\lambda_2'+z\widetilde q_0^{21}	& z\widetilde q_0^{22}&&&&&e^{-2\pi iN(\theta+i\eps)}\lambda_2\\
		z\widetilde q_0^{11}&-z\widetilde q_0^{12}+\lambda_2'&\lambda_2\\
		&\lambda_2&-\lambda_2'+zq_2^{21}&zq_2^{22}&\\
		&&z\widetilde q_2^{11}&-z\widetilde q_2^{12}+\lambda_2'&\lambda_2&\\
		&&&\ddots&\ddots&\ddots&\\
		&&&&&-\lambda_2'+z\widetilde q_{N-1}^{21}&z\widetilde q_{N-1}^{22}\\
		e^{2\pi iN(\theta+i\eps)}\lambda_2&&&&&z\widetilde q_{N-1}^{11}&\lambda_2'-z\widetilde q_{N-1}^{12}
	\end{bmatrix},
	\end{aligned}
\end{equation*}
and hence
\begin{equation*}
	f_N(\theta+i\eps)=\det\left(\widetilde{\mathcal L}_N^\top-z\left(\widetilde{\mathcal M}_N^\top\right)^{-1}\right)\det\left(\widetilde{\mathcal M}_N^\top\right)=(-1)^{N}\det\left(\widetilde{\mathcal L}_N^\top-z\left(\widetilde{\mathcal M}_N^\top\right)^{-1}\right).
\end{equation*}

Let us assume that $\eps>0$. Since $e^{-2\pi N\eps}\to0$ as $N\to\infty$, we ignore the contribution from this term to $f_N$ and hence obtain
\begin{equation}\label{eq:wind_integrand_0}
	f_N(\theta+i\eps)=(-1)^N\Big[(-1)^{2N+1}e^{-2\pi iN(\theta+i\eps)}\lambda_2^{N}z^{N}\prod_{k=0}^{N-1}\widetilde q^{11}_k+\det(G_{z,N})\Big] +\calO\left(e^{-2\pi N\eps}\right).
\end{equation}
Denote by $\chi_\Lambda$ the projection onto $\Lambda\subset\bbZ$, and denote by $\calL\vert_{[0,N-1]}:=\chi_{[0,N-1]}\calL\chi_{[0,N-1]}$ the projection of $\calL$ onto the subset of cells at $0,\dots,N-1$. This $\mathcal L\vert_{[0,N-1]}$ is not unitary but differs from $\mathcal L_N$ by the boundary terms. In contrast, $\calM\vert_{[0,N-1]}=\calM_N$ since $\calM$ is block-diagonal with respect to the cells. With this notation, the matrix $G_{z,N}$ in \eqref{eq:wind_integrand_0} is given by
\begin{align*}
	G_{z,N}&=\widetilde{\mathcal L}^\top|_{[0,N-1]}-z\left(\widetilde{\mathcal M}_N^\top\right)^{-1}\\
	&=
	\begin{bmatrix}
		-\lambda_2'+z\widetilde q_0^{21}	& z\widetilde q_0^{22}&&&&&\\
		z\widetilde q_0^{11}&-z\widetilde q_0^{12}+\lambda_2'&\lambda_2\\
		&\lambda_2&-\lambda_2'+z\widetilde q_2^{21}&&\\
		&&\ddots&\ddots&\ddots&\\
		&&&&&-\lambda_2'+z\widetilde q_{N-1}^{21}&z\widetilde q_{N-1}^{22}\\
		&&&&&z\widetilde q_{N-1}^{11}&\lambda_2'-z\widetilde q_{N-1}^{12}
	\end{bmatrix}
\end{align*}
and describes the dual model with $\theta=0$ and open boundary conditions. Since $\widetilde{\mathcal M}_N$ is unitary we find that
\begin{equation*}
	\det(G_{z,N})=(-1)^N\det\left(\widetilde{\mathcal M}_N^\top\widetilde{\mathcal L}^\top\vert_{[0,N-1]}-z\idty\right)=(-1)^N\det\left(\widetilde{W}\vert_{[0,N-1]}-z\idty\right).
\end{equation*}
Since $G_{z,N}$ does not depend on $\theta$, plugging \eqref{eq:wind_integrand_0} into the integrand in \eqref{eq:winding_num_f} we obtain
\begin{align}
	\frac1{2\pi i}\frac1N\frac{\partial_\theta f_N(\theta)}{f_N(\theta)}
	&=\frac1{2\pi i}\frac1N\frac{-(-1)^{2N+1}2\pi i Ne^{-2\pi iN(\theta+i\eps)}\lambda_2^{N}z^{N}\prod_{k=0}^{N-1}\widetilde q^{11}_k}{(-1)^{2N+1}e^{-2\pi iN(\theta+i\eps)}\lambda_2^{N}z^{N}\prod_{k=0}^{N-1}\widetilde q^{11}_k+(-1)^N\det(G_{z,N})}\nonumber\\
	&=\frac{e^{-2\pi iN(\theta+i\eps)}\lambda_2^{N}z^{N}\prod_{k=0}^{N-1}\widetilde q^{11}_k}{-e^{-2\pi iN(\theta+i\eps)}\lambda_2^{N}z^{N}\prod_{k=0}^{N-1}\widetilde q^{11}_k+(-1)^N\det(G_{z,N})}.\label{eq:wind_integrand_1}
\end{align}

To calculate the limit $N\to\infty$ of the absolute value of this expression, we need to estimate the behaviour of $|\det(G_{z,N})|$ in the limit. To this end, we note that
\begin{lemma}\label{lem:det_Az}
	Let $z\in\partial\bbD$ lie in a gap of the spectrum of $W_{\lambda_1,\lambda_2,0,0}$. Then in the limit $N\to\infty$ we have that
	\begin{equation}\label{eq:char_poly_thouless}
		\lim_{N\to\infty}\frac1{N}\log\left|\det\left(\widetilde{W}\vert_{[0,N-1]}-z\right)\right|=L_{\lambda_1,\lambda_2}^\aubrydual(z)-\log\left|\frac2{\lambda_2(1+\lambda_1')}\right|+\log|z|,
	\end{equation}
	where $L_{\lambda_1,\lambda_2}^\aubrydual(z)=L_{\lambda_1,\lambda_2}^\aubrydual(z)$ denotes the Lyapunov exponent of the Aubry-dual UAMO in the reciprocal setting $\eps=0$.
\end{lemma}

\begin{proof}
	Let $\{z_k:k=1,\dots,2N\}$ be the eigenvalues of $\widetilde W\vert_{[0,N-1]}$ and let $z\in\partial\mathbb D$ lie in a gap of the spectrum of $W_{\lambda_1,\lambda_2,0,0}$. The characteristic polynomial of $\widetilde W\vert_{[0,N-1]}$ is a monic polynomial of order $2N$, such that
	\begin{equation*}
		\log\left|\det\left(\widetilde{W}\vert_{[0,N-1]}-z\right)\right|=\log\left|\prod_{k=1}^{2N}(z_k-z)\right|=\sum_{k=1}^{2N}\log|z_k-z|.
	\end{equation*}
	Defining a probability measure $\widetilde\dos_\Lambda$ for $\Lambda\subset\bbZ$ by $|\Lambda|^{-1}\tr(\chi_\Lambda f(\widetilde W) \chi_\Lambda)=:\int_{\partial\bbD} f(z)  \, d\widetilde\dos_\Lambda(z)$ for any continuous function $f$, we obtain in the limit $N\to\infty$ for $f(\zeta)=\log|\zeta-z|$ that \cite[Proposition 3.1]{joyeDensityStatesThouless2004}
	\begin{equation*}
		\frac1{2N}\sum_{k=1}^{2N}\log|z_k-z|\equiv\int_{\partial\bbD} \log|z'-z| \, d\widetilde\dos_{[1,2N]}(z')	\longrightarrow \int_{\partial\bbD} \log|z'-z|  \,  d\dos(z'),
	\end{equation*}
	where $\dos$ is the density of states measure for $\widetilde W$.
	
	Adapting the Thouless formula from \cite{joyeDensityStatesThouless2004,wernerLocalizationRecurrenceQuantum2013a} from CMV to GECMV matrices, we can relate the density of states measure $\dos$ to the Lyapunov exponent of the dual UAMO for $z\in\bbC$ in the following way:
	\begin{equation*}
		L_{\lambda_1,\lambda_2}^\aubrydual(z)=2\int_{\partial\bbD} \log|z-z'| \, d\dos(z')+\log\left|\frac2{\lambda_2(1+\lambda_1')}\right|-\log|z|.
	\end{equation*}
	Thus, in the limit $N\to\infty$,
	\begin{equation*}
		\frac1{2N}\log\left|\det\left(\widetilde{W}\vert_{[0,N-1]}-z\right)\right|\longrightarrow\frac12L_{\lambda_1,\lambda_2}^\aubrydual(z)-\frac12\log\left|\frac2{\lambda_2(1+\lambda_1')}\right|+\frac12\log|z|.
	\end{equation*}
\end{proof}			
		
Plugging \eqref{eq:char_poly_thouless} into \eqref{eq:wind_integrand_1} and using \cite[Lemma 4.7]{CFO1} to see that for $N$ large enough $\prod_{k=0}^{N-1}\widetilde q^{11}_k=\left[\frac{1+\lambda_1'}{2}\right]^N$ we obtain
\begin{align*}
	\frac1{2\pi i}\frac1N\frac{\partial_\theta f_N(\theta)}{f_N(\theta)}
	&=\frac{e^{-2\pi iN(\theta+i\eps)}\lambda_2^{N}z^{N}\prod_{k=0}^{N-1}\widetilde q^{11}_k}{-e^{-2\pi iN(\theta+i\eps)}\lambda_2^{N}z^{N}\prod_{k=0}^{N-1}\widetilde q^{11}_k+(-1)^{N}e^{NL_{\lambda_1,\lambda_2}^\aubrydual(z)}\left|\frac{\lambda_2(1+\lambda_1')}2\right|^N|z|^N}\\
	&=\frac{e^{-2\pi iN(\theta+i\eps)}}{-e^{-2\pi iN(\theta+i\eps)}+(-1)^{N}e^{NL_{\lambda_1,\lambda_2}^\aubrydual(z)}}	\\
	&=\frac{e^{-2\pi iN(\theta+i\eps)}}{e^{-2\pi iN(\theta+i\eps)}\Big(-1+(-1)^{N}e^{NL_{\lambda_1,\lambda_2}^\aubrydual(z)+2\pi iN(\theta+i\eps)}\Big)}	\\
	&=\frac{1}{-1+(-1)^Ne^{N\left(L_{\lambda_1,\lambda_2}^\aubrydual(z)-2\pi\eps\right)+2\pi iN\theta}}.
\end{align*}
Thus, for $\eps>0$ we have
\begin{equation}\label{eq:wind_app}
	\wind_\eps(z)=\lim_{N\to\infty}\frac1{2\pi i}\frac1N\int_0^1 \frac{\partial_\theta f_N(\theta)}{f_N(\theta)} d\theta=\begin{cases}0&\eps<L^\aubrydual(z)/(2\pi)\\-1&\eps>L^\aubrydual(z)/(2\pi).\end{cases}
\end{equation}	
For $\eps<0$, we have to modify \eqref{eq:wind_integrand_0} accordingly noting that in this case the other corner element vanishes in the large $N$ limit. Following the same steps as above, one sees that the effect of these modifications boil down to a global sign which then proves the assertions of Proposition \ref{prop:wind}.

\end{document}